\documentclass[]{interact}
\DeclareMathAlphabet{\pazocal}{OMS}{zplm}{m}{n}
\usepackage{amsmath}
\usepackage{bbm}
\DeclareMathOperator*{\argmax}{arg\,max}

\usepackage[normalem]{ulem}
\usepackage{amsthm}
\usepackage{upgreek}
\usepackage{empheq}
\usepackage{bm}
\usepackage{epstopdf}
\usepackage{graphicx}
\newif\ifreview
\reviewtrue  
\ifreview

\else
  \usepackage[nolists,tablesfirst]{endfloat}
\fi

\usepackage{hyperref}
\usepackage{tabularx,booktabs}
\usepackage{makecell}
\usepackage{subcaption}
\usepackage{xcolor}

\newcolumntype{C}{>{\centering\arraybackslash}X}
\usepackage[style=authoryear-comp,natbib=true,maxbibnames=99,maxcitenames=2, uniquelist=false,giveninits=true,uniquename=false]{biblatex}
\renewbibmacro*{volume+number+eid}{%
  \printfield{volume}%
  \setunit*{\addnbspace}
  \printfield{number}%
  \setunit{\addcomma\space}%
  \printfield{eid}}
\DeclareFieldFormat[article]{number}{\mkbibparens{#1}}
\renewbibmacro{in:}{%
  \ifentrytype{article}{%
  }{%
    \printtext{\bibstring{in}\intitlepunct}%
  }%
}

\usepackage{tocloft}

\theoremstyle{plain}
\newtheorem{theorem}{Theorem}[section]

\theoremstyle{definition}
\newtheorem{definition}[theorem]{Definition}

\theoremstyle{remark}

\begin{document}

\articletype{}

\title{A CAV-based perimeter-free regional traffic control utilizing existing parking infrastructure}

\author{
\name{Hao Liu\textsuperscript{a}\thanks{CONTACT Hao Liu. Email: hao.liu1@maine.edu}, Vikash V. Gayah\textsuperscript{b}}
\affil{\textsuperscript{a}Department of Civil and Environmental Engineering, University of Maine, USA; \textsuperscript{b}Department of Civil and Environmental Engineering, The Pennsylvania State University, USA}
}

\maketitle

\begin{abstract}
This paper proposes a novel \textit{perimeter-free} regional traffic management strategy for networks under a connected and autonomous vehicle (CAV) environment. The proposed strategy requires a subset of CAVs to temporarily wait at nearby parking facilities when the network is congested. After a designated holding time, these CAVs are allowed to re-enter the network. Doing so helps reduce congestion and improve overall operational efficiency. Unlike traditional perimeter control approaches, the proposed strategy leverages existing parking infrastructure to temporarily hold vehicles in a way that partially avoids local queue accumulation issues. Further, holding the vehicles with the longest remaining travel distances creates a self-reinforcing mechanism which helps reduce congestion more quickly than perimeter metering control. Simulation results show that the proposed strategy not only reduces travel time for vehicles that are not held, but can also reduce travel times for some of the held vehicles as well. Importantly, its performance has been demonstrated under various configurations of parking locations and capacities and CAV penetration rates. 
\end{abstract}

\begin{keywords}
Connected and autonomous vehicles; MFD; Temporary holding; Regional traffic control; Maximum stability
\end{keywords}


\section{Introduction}

Despite substantial research effort in the area of traffic flow management, traffic congestion remains a demanding issue due to the continuous increase in vehicular traffic. In urban networks, traffic signals are one of the most cost effective means to manage traffic flows at intersections, which serve as the primary recurring bottlenecks in urban areas. Based on the spatial scale that is considered when making signal timing decisions, traffic signal control methods can be classified as isolated intersection control (e.g., \citet{wunderlich2008novel, YANG2016109, yu2018integrated, wan2018value}), coordinated signal control along arterial streets (e.g., \citet{bazzan2005distributed, zhang2015band, li2020connected}), and network-wide traffic signal control (e.g., \citet{mirchandani2001real,arel2010reinforcement, varaiya2013max, chen2020toward}). 

Compared to isolated and coordinated arterial level control, a well designed network-wide traffic signal control algorithm can be more beneficial for the operational efficiency of the entire network. However, due to the complexity and computational burden resulting from interdependence between the large number of elements, network-wide control remains a challenging task. Two main categories of network-wide traffic signal control exist based on how decisions at individual intersections are made relative to one another: centralized control (e.g., \citet{lowrie1990scats, robertson1991optimizing, park1999traffic, li2016traffic, liang2019deep}) and decentralized control (e.g., \citet{lammer2008self, varaiya2013max, kouvelas2011hybrid, rinaldi2016sensitivity, sha2019comparative, chen2020toward}). Centralized control methods usually rely on large-scale optimization programs in which signal timings at multiple intersections are determined jointly considering their influence on each other. Although this method can potentially achieve global optimal solutions, they are usually not scalable due to the computational complexity involved. On the contrary, signal timings at each intersection derived from decentralized control methods are independent of other intersections; however, the computational speed of this strategy is usually higher than the centralized control methods. Thanks to the computational efficiency, which is a crucial factor for real-time traffic signal control, more research effort has been shifted to this method in the past decade.

In general, both centralized and decentralized network-wide traffic signal control strategies require estimates of traffic state, such as queue length, at local intersections, which is difficulty to satisfy for most intersections. In addition, most of those algorithms do not consider the relationship between network-level traffic conditions and local signal timings. As a result, they cannot prevent or tackle regional congestion effectively. Perimeter control, which is a method derived based on the theory of Macroscopic Fundamental Diagram (MFD) \citep{godfrey1969mechanism, daganzo2007urban, geroliminis2008existence}, is an effective strategy to tackle both issues. A typical MFD reflects the uni-modal and concave relationship between the average operational efficiency, represented by average flow or vehicle exit rate, and the average use, represented by average density or vehicle accumulation, in a specific region of a network. In general, there exists a critical value for the average use of the network at which the maximum operational efficiency at a network-level can be achieved. Under high demand conditions, the network enters the congested domain once the average use exceeds this critical value, leading to a reduction in the operational efficiency. When traffic demand for the target region is high enough to lead congestion, perimeter control limits the inflows to the target region by adding restrictions at the perimeter of the region to maintain the accumulation or density inside the target region around the critical value, so that the operational efficiency can be near the maximum. Compared to the network-wide traffic signal control algorithms mentioned in the previous paragraph, perimeter metering control does not require the traffic state knowledge at all intersections, and it can result in improvements to overall network-level operations. 

Various perimeter control algorithms have been developed in the literature. For example, Bang-Bang control \citep{daganzo2007urban} allows no vehicle to enter the region whenever the accumulation exceeds the critical value and removes the restriction only once the density drops below the critical value. This strategy has been demonstrated to be an optimal control strategy for abstract systems in which the effect of queue accumulation is not considered. Another widely used perimeter control strategy is proportion-integral-type (PI) feedback control \citep{keyvan2012exploiting, ramezani2015dynamics, haddad2014robust}. Unlike Bang-Bang control, the restriction is adjusted based on the inflows in the previous time step, the change of the average density in the previous time step, and the distance between the current density and the critical density. Therefore, the restriction from this strategy is intensified or loosened more smoothly than Bang-Bang control. Reinforcement learning-based approaches \citep{chu2016large, ni2019cordon, zhou2021model, chen2022data, zhou2023scalable} have also recently been developed and shown to provide reasonable performances compared to MFD theory-based models.

Algorithms that combine perimeter control and decentralized control algorithms have been proposed to further enhance the efficiency of urban transportation systems. For example, \citet{tsitsokas2023two} developed a bi-level control framework to combine PI feedback control and Max Pressure (MP) algorithm \citep{tassiulas1990stability, varaiya2013max}, which is a popular decentralized traffic signal control algorithm, and unveiled that the proposed approach can further enhance the efficiency compared to a purely MP controlled network. \citet{su2023hierarchical} developed a similar approach to integrate reinforcement learning perimeter control algorithm and MP control. In both studies, perimeter control and MP are separately implemented at perimeter intersections and other intersections, respectively. More recently, \citet{liu2024nmp} proposed an algorithm, called N-MP, to incorporate perimeter control mechanism directly within the MP control structure. The N-MP algorithm still follows the decentralized control architecture of traditional MP algorithms and can tackle the heterogeneous vehicle distribution inside the protected region. 

Despite the promising performance of perimeter control algorithms, there exist common drawbacks that can limit their potential. For one, most perimeter control approaches require the identification of perimeter intersections, which are both pre-determined and fixed. However, as pointed out in \citep{li2021perimeter, doig2024and}, a stationary perimeter can impede the full potential of perimeter control when traffic conditions vary in both the temporal and spatial dimensions, which is not an unusual phenomenon for congested regions. More importantly, all perimeter control algorithms lead to temporary queue accumulation at the external boundary of the perimeter. These queues can not only reduce the operation efficiency by blocking destinations at the external region, but they can also keep vehicles from exiting the protected region if spillover occurs and might also contribute to high congestion and gridlock-like conditions near the perimeter. 

Triggered by these drawbacks, this paper considers a \textit{perimeter-free} control strategy to deal with regional congestion. The proposed method assumes that a subset of vehicles are able to receive and will comply with requests to temporarily exit the network and wait at nearby parking facilities in order to improve overall network efficiency. In particular, when the network is congested, vehicles with significant remaining travel distances are requested to use nearby parking spaces to wait. Doing so can help mitigate congestion and reduce travel time, allowing them to exit the network more quickly and reduce overall congestion levels. After a certain time, the held vehicles are allowed to re-enter the network to complete their trips. Importantly, prioritizing long-distance trips for holding creates a self-reinforcing congestion mitigation mechanism. First, temporarily removing these vehicles directly decreases network accumulation. Second, because the remaining vehicles tend to have shorter remaining distances, they can complete their trips more quickly, which increases the vehicle exit rate and further accelerates congestion dissipation.  

We assume for simplicity that the vehicles that are held are Connected and Autonomous Vehicles (CAVs). These CAVs are compliant with the holding requests for two reasons. First, many CAVs—such as delivery vehicles, empty cruising taxis, and fleet-operated vehicles—do not involve human passengers. These vehicles can fully comply with holding instructions without behavioral resistance. Second, even for occupied CAVs, the perception of holding from CAV passengers is different from that from drivers of HDVs. The former can continue working or engaging in other activities while the vehicle is being held, and they do not face the burden associated with re-entering traffic. As a result, their perceived inconvenience is expected to be lower, leading to higher expected compliance compared with human drivers. Note in this study, we assume that CAVs and human-driven vehicles (HDVs) share identical driving behavior parameters; the only distinction is that CAVs comply with holding instructions.

The proposed strategy harnesses both the control flexibility from CAVs and the available parking infrastructure to reduce the negative impact of vehicle accumulation at perimeter boundaries and further improve overall operational efficiency. The results offer insights for policymakers in creating and allocating parking facilities to support the implementation of this strategy. We would like to emphasize that the primary objective of this paper is to explore, for the first time, the benefits and feasibility of the proposed holding strategy as a novel traffic flow management approach, paving the way for relevant future research. Developing detailed theoretical models to enhance the strategy, including the optimization of critical parameters of the proposed method, is beyond the scope of this paper and will be addressed in future studies.

While previous studies have explored various aspects of CAV parking, including motion planning \citep{hsieh2008parking, wang2013hybrid}, parking lot layout  \citep{nourinejad2018designing, yan2022time}, parking pricing \citep{wang2021optimal, fulman2018establishing}, the objective of this paper is to explore the benefits of a proposed holding strategy that utilizes parking facilities in a CAV environment to improve real-time traffic operational efficiency—an area that, to the best of our knowledge, has not yet been studied. Therefore, the literature review focuses on traffic flow management strategies, including traffic signal control and perimeter control, which are popular network-wide techniques for mitigating regional traffic congestion, rather than on parking-related research that addresses other perspectives. One relevant study \citep{xu2017optimal} investigated parking provisions for e-hailing ride-sourcing vehicles (RVs), proposing a parking strategy for vacant RVs to reduce cruising vehicles in the network and thereby alleviate traffic congestion. Our objective, however, differs in that we focus on general private vehicles, including CAVs and potentially HDVs, which leads to different findings regarding the impact of parking provisions on traffic conditions. Overall, our results show that providing parking spaces during peak hours can help reduce congestion, whereas that study concluded that parking provisions for RVs are only necessary during non-peak hours, as cruising is minimal in peak periods.

Any traffic signal control algorithm can be used as the base control for the proposed strategy at intersections within the entire network. In this paper, we evaluate the performance of the proposed strategy integrated with the queue-based MP control developed by \citet{varaiya2013max}, using microscopic traffic simulation. To demonstrate its effectiveness as a network-level traffic management strategy, we compare it against existing perimeter control strategies, which are among the most effective methods for mitigating regional congestion. Simulation results show that the proposed approach can outperform the classical Bang-Bang control strategy and N-MP under various traffic conditions, even under partial CAV environments in which only a subset of vehicles are eligible to be held. Interestingly, the proposed algorithm can reduce travel time for both held and non-held vehicles. Moreover, we proved that if the base signal control has the maximum stability property, it still holds after implementing the proposed holding strategy.

The rest of this paper is organized as follows. Section \ref{sec:method} describes the proposed methodology and proved the maximum stability property. Section \ref{sec: simulation} depicts the microscopic traffic settings, followed by the simulation results shown in Section \ref{sec: results}. The findings are summarized and future research directions are discussed in Section \ref{sec: conclusion}. 
\section{Methodology}\label{sec:method}
We consider a general traffic network $\pazocal{F}=(\pazocal{N}, \pazocal{L})$ where $\pazocal{N}$ is the set of intersections, and $\pazocal{L}$ denote the set of directional links. A link is a pair of adjacent and connected intersections, denoted by $(i,j)$ where $i,j\in \pazocal{N}$. A movement is a pair of incoming and outgoing links at the same intersection serving vehicle transitions. It is represented as a tuple of intersection indices, such as $(i,j,k)$, where $(i,j)$ is an incoming link and $(j,k)$ is an outgoing link at intersection $j$. Let $\pazocal{U}^i$ ($\pazocal{D}^i$) indicate the set of upstream (downstream) intersections of intersection $i$ if $i$ is not an entry (exit) intersection. Let $\pazocal{M}^{i}$ and $\pazocal{S}^i$ indicate the set of movements and the set of admissible phases at intersection $i$, respectively. An admissible phase at intersection $i$ is represented by an array with a length of $|\pazocal{M}^{i}|$, in which each element is a binary variable indicating if the associated movement is served by the phase. Let $\mathbf{S}^{i*}(t)$ be the phase that is activated at intersection $i$ at time $t$ and $S^{i*}_{h,i,j}(t)$ be the element associated with the movement $(h,i,j)$ in $\mathbf{S}^{i*}(t)$.

We assume that all intersections in the network are controlled under a base signal control algorithm.\footnote{Any traffic signal control algorithm can be used as the base control for the proposed strategy at intersections within the entire network. In this paper, we evaluate the performance of the proposed strategy integrated with the queue-based MP control developed by \citet{varaiya2013max}, using microscopic traffic simulation.} We consider both fully CAV environments, where all vehicles comply with the control policy and temporarily park when requested, and mixed HDV-CAV environments, where only CAVs adhere to the parking requests. According to MFD theory, there exists a critical value for the average density, $\rho_{cr}$, at which the network efficiency is maximized. In the proposed algorithm, once the average density, $\rho(t)$, exceeds $\rho_{cr}$, the CAVs with a remaining travel distance longer than a threshold value, $\phi$, are requested to park at a nearby parking space, if a space is available. The decision to hold long-distance trips is justified both theoretically (see Section \ref{sec:mfd holding decision}) and numerically (see Section \ref{sec:randomparking_limit}). After a holding time of $\tau$, the held CAVs will attempt to re-enter the network to finish their trips as soon as there exists available space on the road. Note only vehicles that are already passing the parking locations at the moment when $\rho(t)>\rho_{cr}$ will be held at the associated parking locations, and held vehicles will re-enter the network through the road from which they exit the network. As a result, this strategy does not generate cruising or rerouting maneuvers for parking searching; therefore, the impacts associated with such maneuvers are not considered in this paper.

The rest of this section explains the traffic dynamics and the maximum stability property for the proposed holding strategy. Since a fully CAV environment is a special case of the mixed HDV-CAV environment, we focus on the latter for simplicity. 

\subsection{Traffic Dynamics}\label{sec:dynamics}
At any time $t$, vehicles of movement $(i,j,k)$, i.e., vehicles waiting on link $(i,j)$ to join link $(j,k)$, can be divided into two groups: held CAVs that are currently held at the associated parking location, and vehicles on the streets, called present vehicles. Let $x_{i,j,k}(t)$, $x_{i,j,k}^{\text{CAV,H}}(t)$, $x_{i,j,k}^{\text{CAV,P}}(t)$, and $x_{i,j,k}^{\text{HDV}}(t)$ indicate the total number of vehicles, number of held CAVs, number of present CAVs, and number of HDVs on movement $(i,j,k)$ at time $t$, respectively. Let $y_{i,j,k}(t)$ denote the number of departures of movement $(i,j,k)$ at time $t$. For simplicity, the number of vehicles on all links is modeled using point queues in which the physical length of vehicles is not considered. As a result, all held vehicles can re-enter the network once the held time achieves $\tau$. Then, 
\begin{equation}\label{eq:outflow}
    y_{i,j,k}(t) = \min\left(C_{i,j,k}(t)S^{j*}_{i,j,k}(t),x_{i,j,k}^{\text{CAV,P}}(t)+x_{i,j,k}^{\text{HDV}}(t)\right)
\end{equation}
where $C_{i,j,k}(t)$ is the saturation flow of movement $(i,j,k)$ at time $t$. The exact number of departing CAVs and HDVs, denoted by $y_{i,j,k}^{\text{CAV,P}}(t)$ and $y_{i,j,k}^{\text{HDV}}(t)$, are dependent on the vehicle distribution on link $(i,j)$ and the vehicle transition models. For instance, under a First-In-First-Out (FIFO) principle, the present CAVs and HDVs occupying the first $y_{i,j,k}(t)$ positions in the queue will depart upon receiving a green time. However, our method does not rely on specific queue serving disciplines.  

Let $q^{\text{CAV,P}}_{i,j,k}(t)$ and $q^{\text{HDV}}_{i,j,k}(t)$ indicate the CAV- and HDV-net flows for movement $(i,j,k)$ in time $t$, respectively, i.e., the arrivals from upstream movements subtracted by the departures to downstream movements. Note the transitions between held CAVs and present CAVs are excluded from $q^{\text{CAV,P}}_{i,j,k}(t)$. Then, 
\begin{equation}\label{eq:arrival-departure cav}
    q^{\text{CAV,P}}_{i,j,k}(t)=-y_{i,j,k}^{\text{CAV,P}}(t)+\sum_{h\in {\pazocal{U}^i}}y_{h,i,j}^{\text{CAV,P}}(t)R_{i,j,k}(t),
\end{equation}

\begin{equation}\label{eq:arrival-departure HDV}
    q^{\text{HDV}}_{i,j,k}(t)=-y_{i,j,k}^{\text{HDV}}(t)+\sum_{h\in {\pazocal{U}^i}}y_{h,i,j}^{\text{HDV}}(t)R_{i,j,k}(t),
\end{equation}
where $R_{i,j,k}(t)$ is the turning ratio from link $(i,j)$ to link $(j,k)$.

Then, the evolution of HDVs can be obtained by
\begin{equation}\label{eq: evolution_hdv}
    x_{i,j,k}^{\text{HDV}}(t+1)=x_{i,j,k}^{\text{HDV}}(t)+q^{\text{HDV}}_{i,j,k}(t)
\end{equation}

Since this algorithm holds vehicles only for a predefined threshold duration, the holding times for each held vehicle are tracked. Let $e_v(t)$ indicates the cumulative holding time of held vehicle $v$. Then, the number of holding vehicles at time $t$ that will be allowed to re-enter the network in the next time step is
\begin{equation}\label{eq: enter}
    n^{\text{enter}}_{i,j,k}(t) = \sum_{v=1}^{x_{i,j,k}^{\text{CAV,H}}(t)}\mathbbm{1}(e_v(t)>=\tau)
\end{equation}

Note that $e_v(t)$ accounts for the delay incurred due to both parking and re-entering the network. In other words, the waiting time includes the duration from when a CAV exits the travel lanes to enter an available parking spot to when it re-enters the roadway in the network. Further, while ignored in this section, the simulation results explicitly consider traffic impacts of vehicles exiting/reentering the network before/after their parking maneuvers.

Let $n^{\text{hold}}_{i,j,k}(t)$ denote the number of vehicles that can be potentially held in the next step, i.e., the sum of number of present CAVs and new incoming CAVs with a remaining travel distance longer than $\phi$, if the average density in the network exceeds the critical density, which can be expressed as:
\begin{equation}\label{eq:nhold}
    n^{\text{hold}}_{i,j,k}(t) = \min\left(\mathbbm{1}_{[\rho_{\text{cr}}, +\infty]}\rho(t)\sum_{v=1}^{x_{i,j,k}^{\text{CAV,P}}(t)+q^{\text{CAV,P}}_{i,j,k}(t)}(1-\delta_v)\mathbbm{1}_{[\phi, +\infty]}(g_v(t)),K^p_{i,j,k}-x_{i,j,k}^{\text{CAV,H}}(t)\right)
\end{equation}
where $\delta_v$ denotes a binary variable indicating whether vehicle $v$ has been held, taking the value 1 if held and 0 otherwise; $g_v(t)$ is the remaining travel distance of vehicle $v$ at time $t$, and $K^p_{i,j,k}$ is the total parking capacity associated with movement $(i,j,k)$, which serves as the upper bound of the number of held CAVs from movement $(i,j,k)$. Longer trips are selected for holding because this decision reduces the average remaining travel distance of vehicles circulating in the network, thereby improving overall vehicle exit rate. A detailed theoretical justification is provided in Section \ref{sec:mfd holding decision}.

Note this paper assumes that the layout of parking resources, including the locations and parking capacities, $K^p$'s, is known. The first term in Eq. \eqref{eq:nhold} is the holding demand, and the second term is the available parking supply for movement $(i,j,k)$, respectively.  

Next, the evolution of held CAVs can be expressed as
\begin{subequations}\label{eq:evolution_xh}
	\begin{empheq}[left={x_{i,j,k}^{\text{CAV,H}}(t+1)=\empheqlbrace\,}]{alignat=2}
	& n^{\text{hold}}_{i,j,k}(t), && \quad \text{if $x_{i,j,k}^{\text{CAV,H}}(t)=0$} \label{eq:xh1}\\
	&x_{i,j,k}^{\text{CAV,H}}(t)-n^{\text{enter}}_{i,j,k}(t)+n^{\text{hold}}_{i,j,k}(t) && \quad \text{otherwise}\label{eq:xh2}
	\end{empheq}
\end{subequations}

Similarly, the evolution of present vehicles on movement $(i,j,k)$ can be expressed as,
\begin{subequations}\label{eq:evolution_xp}
	\begin{empheq}[left={x_{i,j,k}^{\text{CAV,P}}(t+1)=\empheqlbrace\,}]{alignat=2}
	& x_{i,j,k}^{\text{CAV,P}}(t)+q^{\text{CAV,P}}_{i,j,k}(t)-n^{\text{hold}}_{i,j,k}(t), && \quad \text{if $x_{i,j,k}^{\text{CAV,H}}(t)=0$} \label{eq:xp1}\\
	&x_{i,j,k}^{\text{CAV,P}}(t)+q^{\text{CAV,P}}_{i,j,k}(t)+n^{\text{enter}}_{i,j,k}(t)-n^{\text{hold}}_{i,j,k}(t) && \quad \text{otherwise}\label{eq:xp2}
	\end{empheq}
\end{subequations}

We assume that CAVs cannot be held at the time step within which they enter a new road. Therefore, $q^{\text{CAV,P}}_{i,j,k}(t)$ only contributes to the evolution of the present CAVs while it does not influence the the held CAVs in the next step. Combining Eqs. \eqref{eq:evolution_xh} and \eqref{eq:evolution_xp} shows that the sum of number of CAVs in both groups always converges to the total number of CAVs. 

\subsection{Theoretical justification for holding long-remaining-distance trips}\label{sec:mfd holding decision}
The decision to hold vehicles with long remaining travel distances provides both direct and indirect performance benefits. The direct benefit arises from reducing network accumulation during congested conditions. Under the unimodal MFD relationship, decreasing accumulation in the congested regime increases network flow, thereby improving system performance. Beyond this accumulation effect, additional benefits emerge from selectively holding vehicles with large remaining travel distances. According to MFD theory, the average network flow can be expressed as the total vehicle-distance traveled per unit time:
\begin{equation}
    Q=\frac{\sum_{v\in \eta}s_vl_v}{L}
\end{equation}
where $\eta$ is the set of vehicles on the network, $s_v$ is the speed of vehicle $v$, $l_v$ is a characteristic link length, and $L$ is the total network length. Equivalently, drawing on Little's Law, the average trip completion rate $\lambda$ relates to accumulation $n$ and mean remaining travel distance $\Tilde{d}$ as: 
\begin{equation}
    n=\lambda\frac{\Tilde{d}}{\Tilde{s}}
\end{equation}
or more compactly, since the average travel speed $\Tilde{s}$ can be solely determined by accumulation, assuming traffic distribution is homogeneous, the trip completion rate can be expressed as:
\begin{equation}\label{eq:nef}
    \lambda=\frac{Q(n)}{\Tilde{d}}
\end{equation}

Equation \ref{eq:nef} provides a key insight: for a given accumulation level $n$ (and thus a given flow $Q(n)$ in the congested regime), the trip completion rate $\lambda$ is inversely proportional to the mean remaining travel distance $\Tilde{d}$. Selectively holding vehicles with long remaining travel distances reduces $\Tilde{d}$ among vehicles circulating in the network. This increases the trip completion rate, which subsequently reduces accumulation and further increases $Q(n)$ through the MFD relationship. The proposed strategy creates a self-reinforcing recovery process: holding long-distance vehicles reduces $\Tilde{d}$, which in turn increases trip completion rate, which further reduces n and increases $Q(n)$. Consequently, holding long-distance trips accelerates congestion dissipation more effectively than indiscriminate holding or holding short-distance trips. The performance of this design choice is demonstrated in Section \ref{sec:randomparking_limit}.

\subsection{Maximum stability property}

Network throughput is an indicator widely used to assess the performance of traffic control approaches. A control algorithm has the \emph{maximum stability} property, which is a favorable property investigated extensively for adaptive traffic signal control methods \citep{varaiya2013max, smith2019traffic}, if it maximizes the network throughput, or equivalently, if it can accommodate all feasible demand patterns. This section proves that if the maximum stability property holds for the base control algorithm, it is preserved by the proposed holding strategy. 
\begin{definition}\label{def: stabilize}
A control policy is said to stabilize the network under demand $\mathbf{d}$ or accommodate demand $\mathbf{d}$ for the network, if the average number of vehicles in the network over any period is upper bounded, i.e.,

\begin{equation}\label{eq: net stable}
    \frac{1}{T}\sum_{t=1}^{T}\sum_{(i,j,k)}\mathop{\mathbb{E}}\left[x_{(i,j,k)}(t)\right]< M, \quad \forall T\in \mathbb{Z^+},
\end{equation}
where $0<M<\infty$ is a constant. 
\end{definition}

Note that, since $x_{(i,j,k)}(t)>0$, Eq. \eqref{eq: net stable} implies that when the network is stabilized, the number of vehicles in the network is always upper bounded. In addition, under the assumption of point queues, which has been commonly employed as an assumption for the proof of maximum stability property (e.g., \citet{varaiya2013max, smith2019traffic}), spillover effect is not considered, and the stabilization is independent of the initial state of the network. 

\begin{definition}
    A demand $\mathbf{d}$ is admissible if there exists a signal control strategy that can stabilize the network under $\mathbf{d}$.
\end{definition} 

\begin{definition}
    A signal control policy $\pazocal{P}$ has the maximum stability property if it can stabilize the network under any $\mathbf{d}\in \mathbf{D}_f$, where  $\mathbf{D}_f$ is the set of admissible demand.
\end{definition}

\begin{theorem}
If policy $\pazocal{P}$ has the maximum stability and the maximum number of holding vehicles $x_{\text{max}}^{\text{CAV,H}}$ in the network is bounded, the incorporation of the proposed holding strategy into $\pazocal{P}$, denoted by $\pazocal{P}^\textsuperscript{H}$, has the maximum stability property.
\end{theorem}
\begin{proof}
Assume that at time $t$, the total number of vehicles in the network is upper bounded, i.e.,
\begin{equation}
    \exists M\in(0, \infty), \quad \sum_{(i,j,k)}x_{i,j,k}(t)<M
\end{equation}

Let $\mathbf{N}^{\text{enter}}(t)$ and $\mathbf{N}^{\text{hold}}(t)$ be the vector of number of re-entering and request-to-hold CAVs, defined by Eq. \eqref{eq: enter} and Eq. \eqref{eq:nhold}, respectively, of all movements. Let $\mathbf{X}^{P}(t)$ indicate the total number of present vehicles, including both CAVs and HDVs, in the network, i.e., $\mathbf{X}^{P}(t)=\mathbf{X}^{\text{CAV,P}}(t)+\mathbf{X}^{\text{HDV}}(t)$.

It is reasonable to assume that the control decision made by the base control policy $\pazocal{P}$ at time $t$ is determined based on the traffic state $\mathbf{X}(t)$ and the demand $\mathbf{d}(t)$. Since the original control policy $\pazocal{P}$ only considers present vehicles for signal timings, the conditional expectation of the number of present vehicles in the network at time $t+1$ can be expressed as,
\begin{equation}\label{eq:proof}
\mathbb{E}^{\pazocal{P}^H}\left[\mathbf{1}^\top \mathbf{x}^P(t+1)|\mathbf{X}(t), \mathbf{d}(t)\right]=\mathbb{E}^{\pazocal{P}}\left[\mathbf{1}^\top \mathbf{x}^P(t+1)|\mathbf{X}^P(t)+\mathbf{N}^{\text{enter}}(t)-\mathbf{N}^{\text{hold}}(t), \mathbf{d}(t)\right]
\end{equation}

Although the actual values of $\mathbf{N}^{\text{enter}}(t)$ and $\mathbf{N}^{\text{hold}}(t)$ depends on the traffic condition at time $t$, they are always bounded since the parking capacity is assumed to be limited, indicated by Eq. \eqref{eq:nhold}. Since control policy $\pazocal{P}$ has the maximum stability property, according to Definition \ref{def: stabilize}, for any $\mathbf{d}\in \mathbf{D}_f$, the unconditional expectation of present vehicles from Eq. \eqref{eq:proof} is bounded by a constant $M'$. Additionally, because the total number of held vehicles is always upper bounded by $\sum_{i,j,k}K^p_{i,j,k}$, the expectation of total number of vehicles at time $t+1$ is bounded by $M=M'+\sum_{i,j,k}K^p_{i,j,k}$.

Above all, if $\mathbf{1}^\top \mathbf{x}(1)$ is bounded, the evolution of number of vehicles in the network resulting from the control policy $\pazocal{P}^\textsuperscript{H}$ satisfies \eqref{eq: net stable} for any $\mathbf{d}\in \mathbf{D}_f$. As a result, policy $\pazocal{P}^\textsuperscript{H}$ has the maximum stability property. 

\end{proof}

\section{Numerical Experiment Design}\label{sec: simulation}

The performance of the proposed control strategy is tested in microscopic traffic simulation using the AIMSUN software. Although the proposed holding strategy is independent of the base traffic signal control strategy, to maintain maximum stability in our tests, we consider the case where all intersections are controlled by the queue-based Max Pressure algorithm proposed by \citet{varaiya2013max}, Q-MP, which has the maximum stability property. The holding strategy proposed in the previous section is then applied on Q-MP, called Q-MP\textsuperscript{H}, to show the superiority of the proposed approach on the traffic operational efficiency at a network-level. Both average vehicle delay and MFD --- represented by the relationship between vehicle exit rate and the average density in the network --- are used to quantify the effectiveness. Since the proposed algorithm introduces a novel paradigm for network-wide traffic management, perimeter control, which has been widely studied for tackling regional congestion, is employed as the baseline method. Specifically, Q-MP\textsuperscript{H} is compared to two perimeter control algorithms: Bang-Bang control and N-MP. The performance of the proposed strategy is investigated in both a pure CAV environment and a mixed CAV-HDV traffic environment. The rest of this section describes the simulation settings and the two benchmark perimeter control algorithms. 

\subsection{Simulation settings}\label{sec: net setup}
We consider here a 10x10 square grid network of two-way streets; see Figure \ref{fig:net_sim}. The grid network layout has been widely utilized in traffic flow management for urban transportation systems \citep{chu2016large, tan2019cooperative,chu2016large}. The length of all links is set to 200 m. Since Q-MP requires queue lengths for all movements as the input to determine signal timings, we assume that all links have three dedicated lanes, one for each of the left-turn, through, and right-turn movements. The free flow travel speed is 50 km/h, and the saturation for each link is 1,800 vphpl.

\begin{figure}[!htbp]
	\centering
	\includegraphics[width=0.8\textwidth]{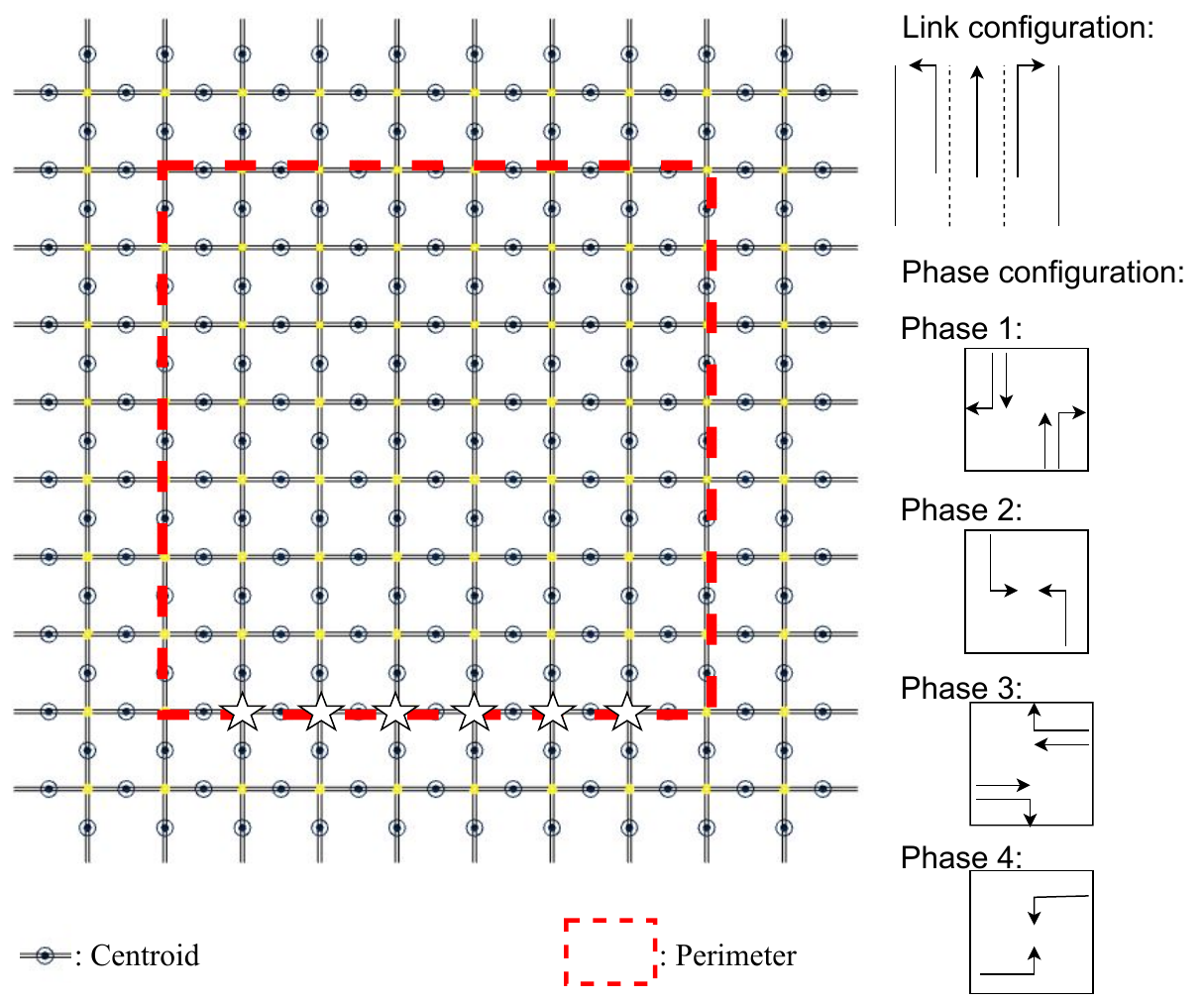}
	\caption{Network setup.}
	\label{fig:net_sim}
\end{figure}

The typical four-phase plan shown in Figure \ref{fig:net_sim} is employed at all intersections. Q-MP updates signal phases every 10 s at all intersections. The yellow time and all-red time are 3 s and 1 s, respectively. 

The simulation time is 2 hours. The locations of centroids that both generate (origin) and attract (destination) vehicle trips are shown in Figure \ref{fig:net_sim}. Two demand patterns are tested: 

\begin{enumerate}
	\item Uniform demand pattern: The expected number of trips between each OD pair in this pattern is identical. It is used to evaluate the improvement in traffic mobility of Q-MP\textsuperscript{H}. To obtain the typical congestion formation-dissipation process, the demand in the first hour is relatively high while the second hour with zero demand serves as a cool down period. After testing multiple demand levels, an hourly rate of 1.05 vph is used as the vehicle generation rate for all ODs. This demand level is high enough to lead to congestion to the network during simulation, which is the condition that the proposed strategy and perimeter control are designed for; however, it is low enough so that all trips can be finished within the simulation time for most scenarios.
	\item Concentrated demand pattern: In this scenario, travel demand is still generated uniformly throughout the entire network, but all vehicles have destinations within a designated central area. This pattern differs from the uniform demand pattern by removing all trips heading to the outer region. Since perimeter control is designed to mitigate congestion caused by regional high demand, this pattern effectively represents such a scenario and has been widely used in perimeter control studies \citep{ni2019cordon, liu2024nmp}. Additionally, simulation results indicate that Bang-Bang control performs worse than pure Q-MP control under the uniform demand pattern. Therefore, this pattern is only utilized in Section \ref{sec:comparison_perimeter} for a fair comparison between Q-MP\textsuperscript{H} and the perimeter control methods.
\end{enumerate}

Five replications with different random starting seeds are run for each scenario to mitigate the influence of randomness. Stochastic c-logit route choice model is used in AIMSUN to emulate user-equilibrium routing conditions in which vehicles make routing decisions to minimize their own personal travel times.

Simulation data is retrieved every second. Then, data is averaged every 100 s to obtain average density (veh/lane-km) and hourly trip completion rate (vph), also called network exit function (NEF), respectively, which are used to represent MFD. 

\subsection{Base control: Q-MP}

Thanks to its ease of implementation and strong operational performance, MP has become a very popular decentralized traffic signal control algorithm \citep{wongpiromsarn2012distributed, varaiya2013max, kouvelas2014maximum, xiao2014pressure, wu2017delay,li2019position, dixit2020simple, wang2022learning,  mercader2020max, noaeen2021real, liu2022novel, liu2023total, smith2024backpressure, liu2024max, ahmed2024occ} in the past decade. A comprehensive review of MP algorithms can be found in \citep{levin2023max}.

The Q-MP (queue-based MP) algorithm proposed by \citet{varaiya2013max} is utilized as the base control for all intersections in the network. To make it easier to follow the proposed method, Q-MP is briefly explained. In Q-MP, the temporal dimension is divided into intervals with fixed duration for signal timing update. At the beginning of each interval, the weight of all movements is calculated as follows: 
\begin{equation}\label{eq:weight_qmp}
	w_{i,j,k}(t)=x_{i,j,k}(t)-\sum_{l\in {\pazocal{D}^k}}x_{j,k,l}(t)R_{j,k,l}(t) \quad \forall (i,j,k)
\end{equation}

Then, the pressure of all phases is calculated by:
\begin{equation}\label{eq:pressure_qmp}
	p(\mathbf{S}^j)(t)=\sum_{(i,j,k)\in \pazocal{M}^{j}}w_{i,j,k}(t)C_{i,j,k}(t)S^j_{i,j,k}(t)\quad \forall \mathbf{S}^j\in \pazocal{S}^j
\end{equation}

At each intersection, the phase with the maximum pressure is activated for the next time step,
\begin{equation}\label{eq:opt}
	\mathbf{S}^{j*}(t)=\argmax_{\mathbf{S}^{j}\in \pazocal{S}^j}p(\mathbf{S}^{j})(t)
\end{equation}

When implementing Q-MP\textsuperscript{H}, at each signal update instant, a subset of CAVs are held and a subset of held CAVs re-enter the network according to the proposed criteria. Each intersection is controlled by Q-MP in which the total number of vehicles on each movement, $x_{i,j,k}(t)$ in Eq. \eqref{eq:weight_qmp}, is replaced with the total number of present vehicles on each movement including both HDVs and present CAVs, i.e., $x_{i,j,k}^{\text{CAV,P}}(t)+x_{i,j,k}^{\text{HDV}}$.

\subsection{Benchmark perimeter control}\label{sec: perimeter base}
Traditional perimeter control algorithms require the identification of a perimeter to limit inflows traversing the perimeter to the congested region. An example of perimeter is shown in Figure \ref{fig:net_sim}. Once the perimeter is identified, a general way to limit inflows is to reduce the green time ratio for phases sending vehicles into the congested region at perimeter intersections, i.e., intersections on the perimeter. Two perimeter control approaches, Bang-Bang control \citep{daganzo2007urban} and N-MP \citep{liu2024nmp}, are selected as the benchmark perimeter control methods. Both approaches will be explained in more details in the rest of this section. Note that \citet{liu2024nmp} has showed that the feedback-based perimeter control \citep{keyvan2012exploiting} generated worse performance than Bang-Bang control under similar settings. Therefore, the former is excluded from this paper.

Under the uniform demand pattern defined in Section \ref{sec: net setup}, Bang-Bang control results in a worse performance than the baseline pure Q-MP control strategy. Therefore, to ensure a fair comparison, the concentrated demand pattern, which has been utilized to demonstrate the performance of perimeter control methods \citep{ni2019cordon}, is employed. Note unlike Bang-Bang control, the performance of the proposed holding control strategy does not rely on this specific demand pattern. Therefore, this modification is only applied in Section \ref{sec:comparison_perimeter}, where to ensure a fair comparison between Q-MP\textsuperscript{H} and the two tested perimeter control approaches.   

\subsubsection{Bang-Bang control}

Bang-Bang control \citep{daganzo2007urban} is a simple but effective perimeter control strategy, which maximizes the operational efficiency for abstract systems. For this benchmark control strategy, all intersections except for perimeter intersections are still controlled by Q-MP, and Bang-Bang control is incorporated into the Q-MP structure at perimeter intersections only using the following method. At each signal update instant, if the density inside the perimeter, $\rho^p(t)$, exceeds the critical density, $\rho^p_{cr}$, all inbound lanes at the perimeter intersections will be blocked by adding red time meters, and the signal at perimeter intersections is updated based on Q-MP without considering those inbound lanes. For example, when $\rho^p(t)>\rho^p_{cr}$, all middle lanes of the northbound approaches at the south-side perimeter intersections (marked by the star symbol in Figure \ref{fig:net_sim}) will be blocked and ignored by the weight calculation in Q-MP. If $\rho^p(t)<\rho^p_{cr}$, perimeter intersections are controlled by the original Q-MP.

\subsubsection{N-MP}
\citet{liu2024nmp} recently proposed an algorithm, called N-MP, to incorporate regional perimeter control mechanism directly into MP control structure, while maintaining its decentralized control structure and the maximum stability property. The major difference between N-MP and traditional MP algorithms is that N-MP considers both local and regional traffic states for the weight calculation for the inbound movements at perimeter intersections. When congestion occurs, the weight for the inbound movements decreases with the increase in the density inside the perimeter, $\rho^p$, and the increase in the upstream queue length of the local movement. Specifically, the weight for an inbound movement in N-MP is expressed as,

\begin{equation}\label{eq:nmp}
	w^N_{i,j,k}(t)=w_{i,j,k}(t)-\Psi\left(\rho^p(t)-\rho^p_{cr}, \;x_{i,j,k}(t)\right),
\end{equation}
where $w_{i,j,k}(t)$ is the weight of Q-MP, defined in Eq. \ref{eq:weight_qmp}, and $\Psi$ is a positive and increasing function of the number of vehicles on the corresponding upstream link, $x_{i,j,k}(t)$, and the exceeded density, $\rho^p(t)-\rho^p_{cr}$, in the protected region. Note Eq. \eqref{eq:nmp} is applied on inbound movements only if $\rho^p(t)-\rho^p_{cr}>0$. After obtaining the weight for all movements, the pressure calculation and phase selection in N-MP are the same as Q-MP. Unlike Bang-Bang control, which generates long queue accumulation resulting from blocking all inbound movements, N-MP determines whether to block an inbound movement based on both local and regional traffic states. As a result, N-MP can offer more efficient control decisions for the regional operational performance. A transformed sigmoid function is utilized for $\Psi$ to maintain the maximum stability property. Specifically,
\begin{equation}\label{eq:psi}
	\Psi\left(\rho^p(t)-\rho^p_{cr}, \;x_{i,j,k}(t)\right)=\xi (\rho^p(t)-\rho^p_{cr})^2\left(\frac{1}{1+\exp{\left(\frac{-x_{i,j,k}(t)}{400}\right)}}-\frac{1}{2}\right)\times10^3,
\end{equation}

\section{Simulation Results}\label{sec: results}
The performance of the Q-MP\textsuperscript{H} is evaluated in both fully CAV environments -- where all vehicles can be requested to hold (Sections \ref{sec: unrestricted parking}-\ref{sec: results restricted parking}) -- to assess the benefits of the proposed strategy under a complete CAV deployment, and in mixed HDV-CAV environments (Section \ref{sec: partial cav}) -- where only a subset of vehicles can be held -- to explore the achievable benefits during the transition era. In addition, to assess the feasibility of the proposed methods, various configurations of parking facilities—including number, distribution, and capacity—are tested in the simulation. Moreover, the effectiveness of holding long-distance trips is further evaluated under a more realistic setting characterized by irregular spatial parking distributions and limited parking capacity. This setup allows us to control for differences in the number of vehicles removed from the network, thereby isolating the effect of the holding rule itself. The network exit rate, a commonly used metric for evaluating regional traffic management performance, is used as the primary performance index. Furthermore, we examine the impact of the proposed method on vehicle delays for both held and non-held vehicles. The scenarios, in terms of the spatial distribution and capacity of parking facilities, in the rest of this section are summarized below:
\begin{itemize}
    \item Section \ref{sec: unrestricted parking}: parking available to all links with infinite parking capacity
    \item Section \ref{sec: results restricted parking}: restricted parking space
    \begin{itemize}
        \item Section \ref{sec: perimeter parking dist}: perimeter-like parking distribution with infinite parking capacity
        \item Section \ref{sec: results with parking limit}: perimeter-like parking distribution with limited parking capacity
        \item Section \ref{sec:comparison_perimeter}: perimeter-like parking distribution with infinite parking capacity
        \item Section \ref{sec:randomparking}: random parking distribution with infinite parking capacity
        \item Section \ref{sec:randomparking_limit}: random parking distribution with limited parking capacity
    \end{itemize}
    \item Section \ref{sec: partial cav}: random parking distribution with infinite parking capacity under varying CAV penetration rates
\end{itemize}

\subsection{Results with unrestricted parking space}\label{sec: unrestricted parking}
This section demonstrates the performance of the proposed strategy assuming that parking spaces are available next to all links, as shown in Figure \ref{fig:net_sim}, with enough capacity so that CAVs can be held when and where needed. This scenario is employed to demonstrate the efficiency of the proposed strategy under an idealized setting. Note that CAVs will be held at the centroid from the same street on which they are currently traveling.

\subsubsection{Comparison to pure Q-MP control}
This section compares Q-MP\textsuperscript{H} to Q-MP without perimeter control. According to Eq. \eqref{eq:nhold}, Q-MP\textsuperscript{H} uses the critical density in MFD, $\rho_{cr}$, to determine whether the holding strategy should be executed. Therefore, we first run simulation in which all intersections are controlled under Q-MP to obtain $\rho_{cr}$. This also serves as the baseline to demonstrate the control performance of the proposed algorithm. 

The MFD for Q-MP control, represented by the relationship between average density and trip completion rate, is shown in Figure \ref{fig:nef density}. Under the uniform demand patterned described in the previous section, the trip completion rate gradually increases from zero to the maximum. Then, after passing the first local maximum point, it starts decreasing, meaning the network enters the congested domain. 

\begin{figure}[h]
	\centering
	\includegraphics[width=0.6\textwidth]{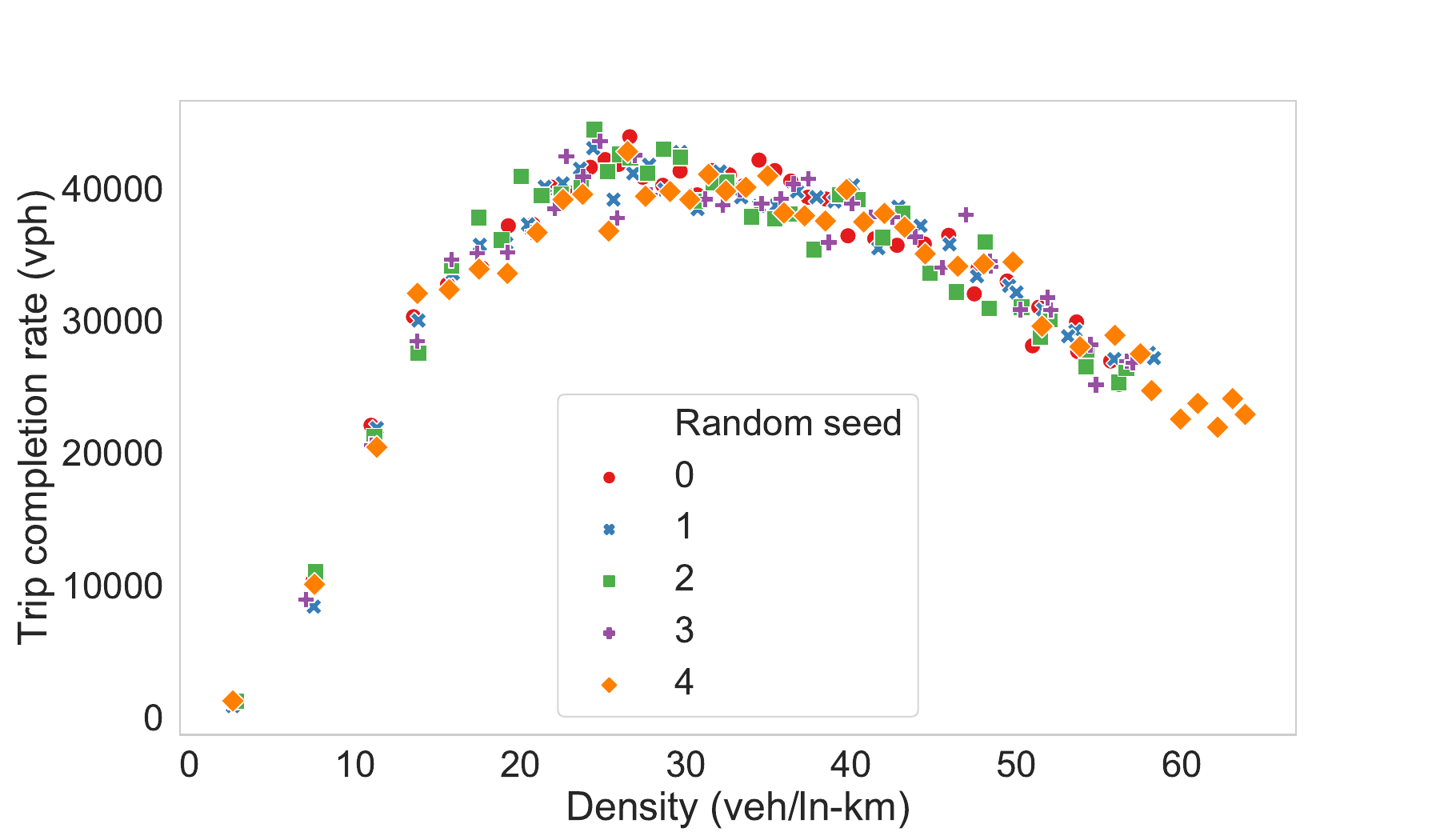}
	\caption{MFD under Q-MP.}
	\label{fig:nef density}
\end{figure}

Figure \ref{fig:nef density} shows that the trip completion rate reaches its maximum when the average density is in the range between 15 veh/lane-km and 40 veh/lane-km. Next, we identify a value from the set of \{15, 20, 25, 30, 35, 40\} that can achieve the highest efficiency for Q-MP\textsuperscript{H} as the critical value, $\rho_{cr}$. For this purpose, we hold all CAVs with a remaining distance of 10 links, i.e., $\phi=10$, once the average density exceeds the tested value. The held vehicles will re-enter the network after 5 minutes. Note that while a CAV can ultimately be held only once in the proposed strategy, the purpose of this section is to determine the critical density that yields optimal control performance. Therefore, in these simulations, CAVs may be held multiple times, meaning their total holding time is not restricted. The average MFDs across five random seeds for all tested critical density are shown in Figure \ref{fig:critical density}. For a better visualization, the vertical axis shows the increase in the cumulative number of exit vehicles from Q-MP\textsuperscript{H} compared to Q-MP. First, it shows that Q-MP\textsuperscript{H} can significantly improve the operational efficiency compared to the pure Q-MP control. Second, with the increase in $\rho_{cr}$, the vehicle exit rate first increases and then decreases, and the best performance is achieved when $\rho_{cr}=20$ veh/lane-km. This value is used for the simulation in the rest of this section.

\begin{figure}[h]
	\centering
	\includegraphics[width=0.6\textwidth]{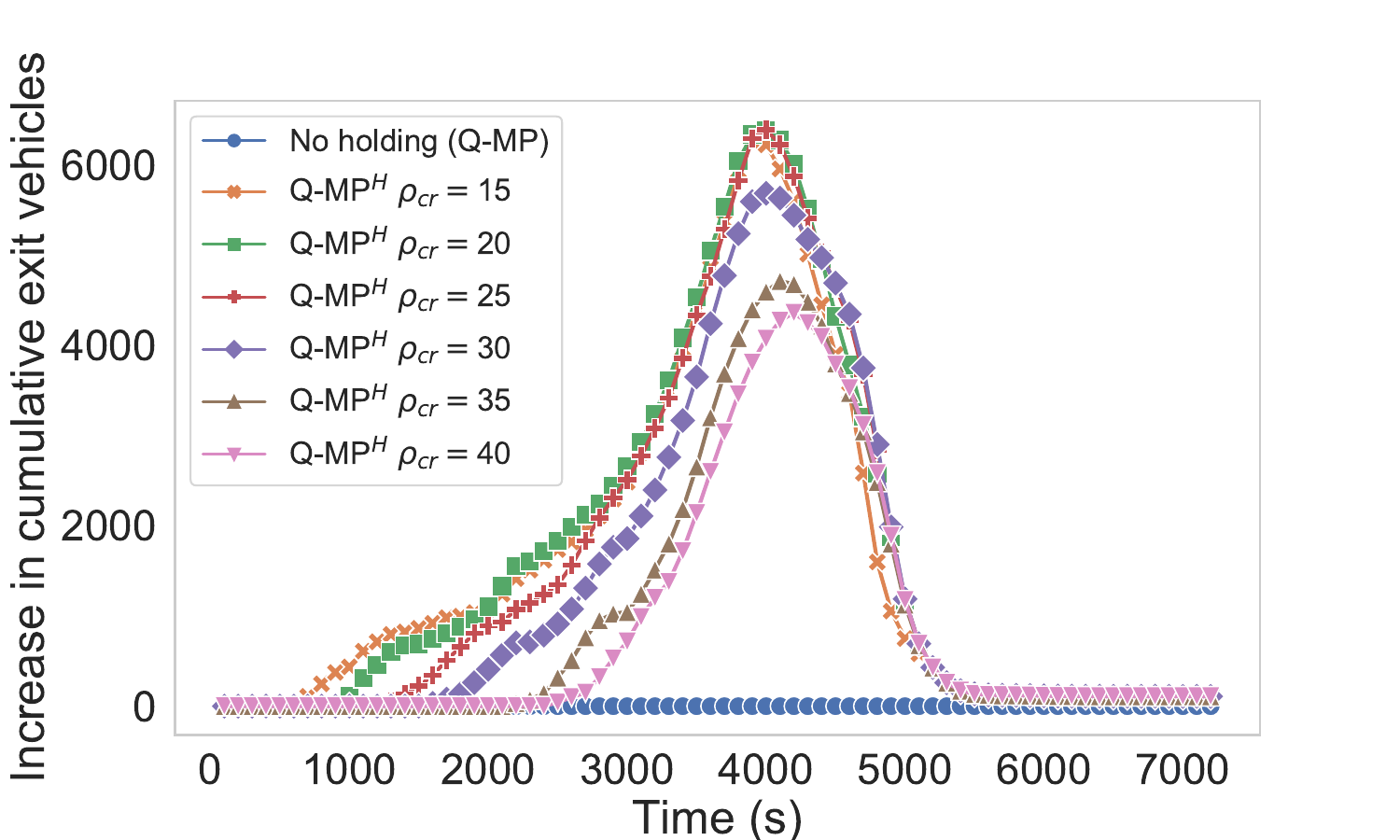}
	\caption{Impact of critical density on network exit rate.}
	\label{fig:critical density}
\end{figure}

In addition to MFD which represents the network-level efficiency, we divide all simulated vehicles into groups based on their total travel distance and examine the change in travel delay for each vehicle group to obtain a higher resolution for the
control performance. Figures \ref{fig:number of trips} and \ref{fig:delay reduction in each group} show the number of vehicles and the total reduction of travel delay in each group, respectively. Under the assumed simulation settings, the travel distance of most trips is shorter than 20 links, and the travel distance of 7 links owns the largest number of vehicle trips. Figure \ref{fig:delay reduction in each group} shows that the proposed holding strategy can reduce travel delay for all vehicle groups including, perhaps surprisingly, the group from which CAVs are temporarily held.

\begin{figure}[!ht]
	\centering
	\begin{subfigure}[h]{0.49\textwidth}
		\centering
		\includegraphics[width=\textwidth]{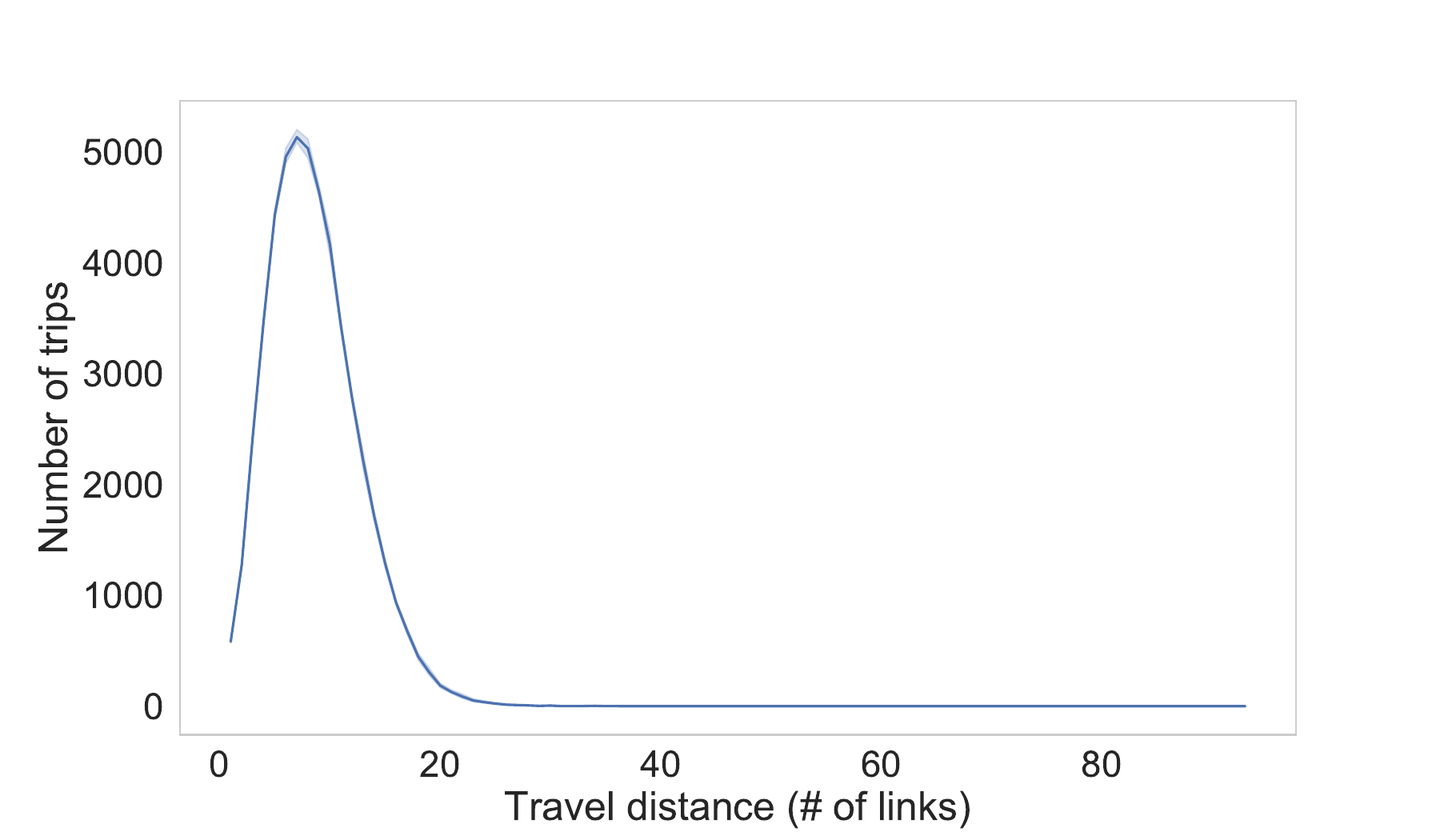}
		\caption{Number of trips in each group}
		\label{fig:number of trips}
	\end{subfigure}
	\begin{subfigure}[h]{0.49\textwidth}
		\centering
		\includegraphics[width=\textwidth]{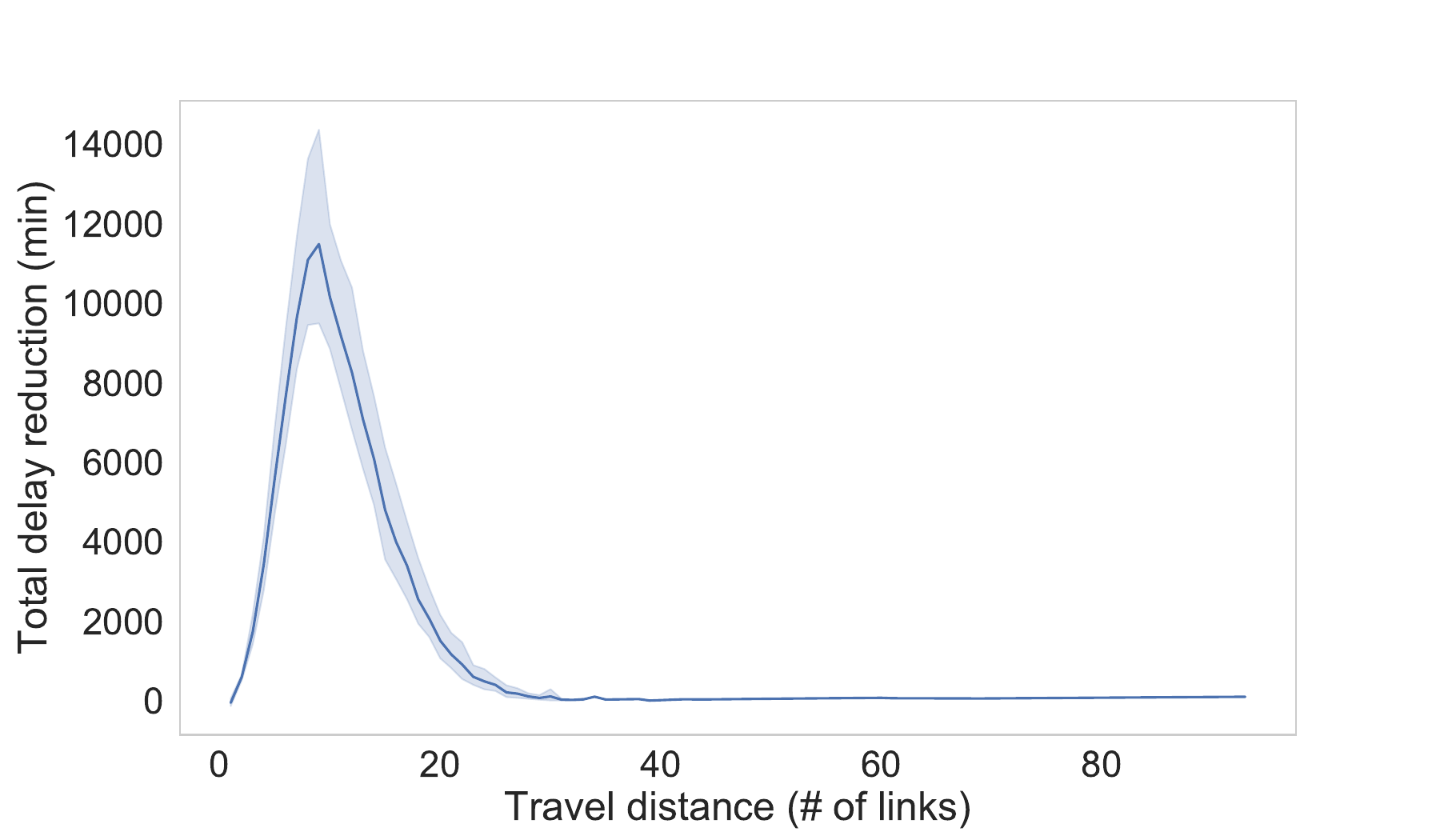}
		\caption{Delay reduction in each group}
		\label{fig:delay reduction in each group}
	\end{subfigure}
	\hfill
	\caption{Performance for vehicle groups.}
	\label{fig:vehicle group}
\end{figure}

Since dynamic routing is allowed in the simulation, it is possible that a subset of vehicles switch to alternative routes with shorter travel distance, if the travel time on such routes becomes shorter than the original routes in the pure Q-MP-based control scenario after implementing the proposed holding strategy. To ensure the improvement in the overall efficiency does not result from a reduction in travel distance, Figure \ref{fig:travel dis compare} shows the comparison of total travel distance between Q-MP and Q-MP\textsuperscript{H}. Although the total travel distance from the proposed algorithm is slightly lower than that from Q-MP under random seeds 1 and 4, the difference is negligible compared to the difference in travel delay, shown in Figure \ref{fig:delay reduction in each group}. Therefore, the reduction in vehicle travel time results from the improvement of traffic conditions rather than the shortening of travel distance.
\begin{figure}[h]
	\centering
	\includegraphics[width=0.5\textwidth]{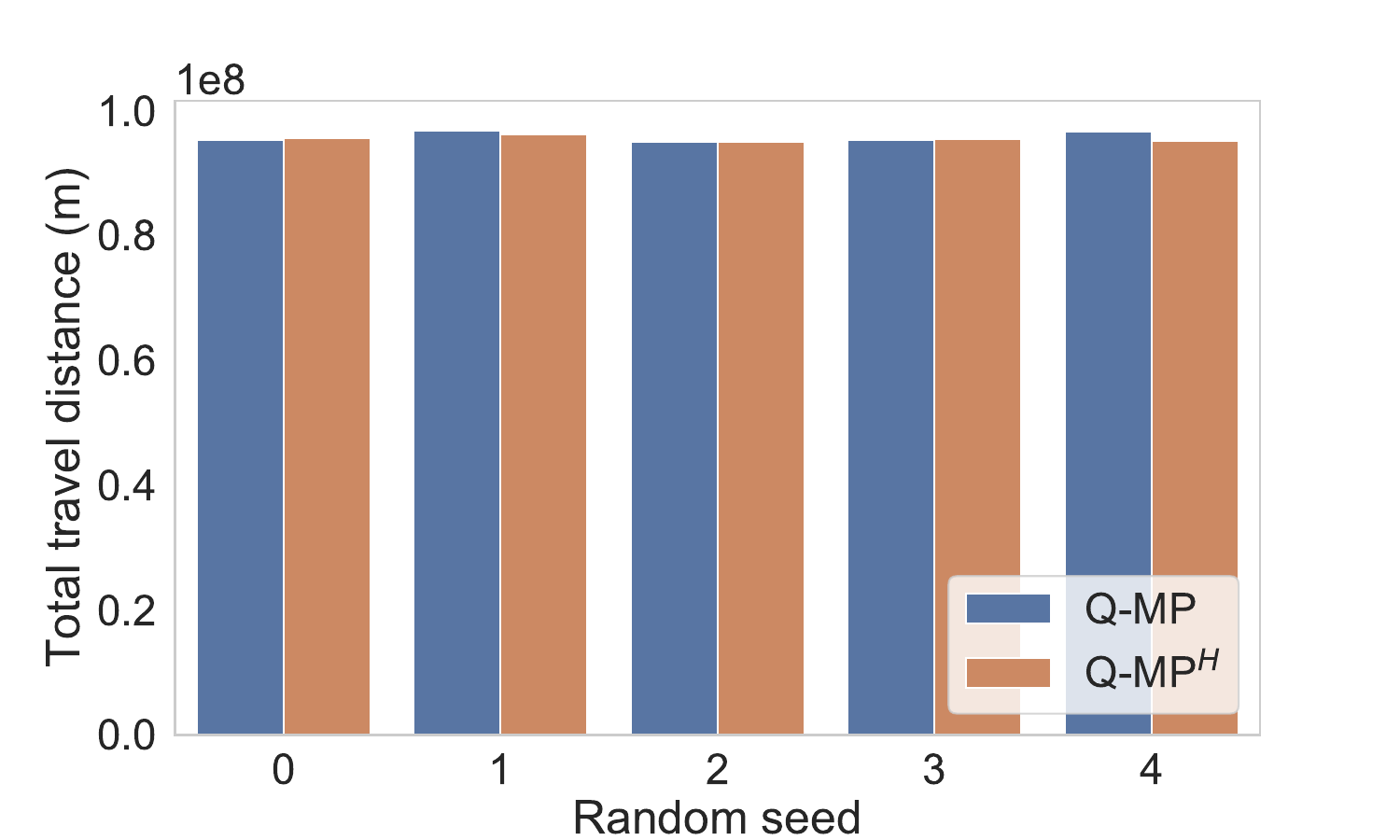}
	\caption{Comparison of travel distance.}
	\label{fig:travel dis compare}
\end{figure}

This section demonstrated the potential of the proposed strategy to improve the operational efficiency at a network level. In the following, we will investigate the impact of the thresholds for remaining travel distance $\phi$ and holding time $\tau$ on the performance and delve deeper into the travel delays associated with travel distance.

\subsubsection{Impact of holding parameters}\label{sec:phi and tau}
We first investigate the influence of $\phi$ on the control performance. Assuming the holding time is set to be 5 mins, five values in the set \{6, 8, 10, 12, 14\} (links) are tested for $\phi$. Same as the previous section, a CAV can be held multiple times. Figure \ref{fig:impact of td} shows the results on both cumulative number of exit vehicles and vehicle delay. Again, for the purpose of visualization, the results from Q-MP are used as baseline in Figure \ref{fig:impact of td}. When $\phi$ is small, an unnecessarily large number of CAVs are held, leading to a low improvement in the efficiency. On the other hand, when $\phi$ exceeds a certain value, the number of held vehicles is not large enough to produce the best traffic conditions, leading to a reduction in the control performance. According to Figure \ref{fig:impact of td}, $\phi=8$ leads to the best control performance. In addition, the marginal improvement in the control performance when $\phi$ increases from 6 to 8 is much more significant than the marginal decline when $\phi$ exceeds 8. 

\begin{figure}[!ht]
	\centering
	\begin{subfigure}[h]{0.49\textwidth}
		\centering
		\includegraphics[width=\textwidth]{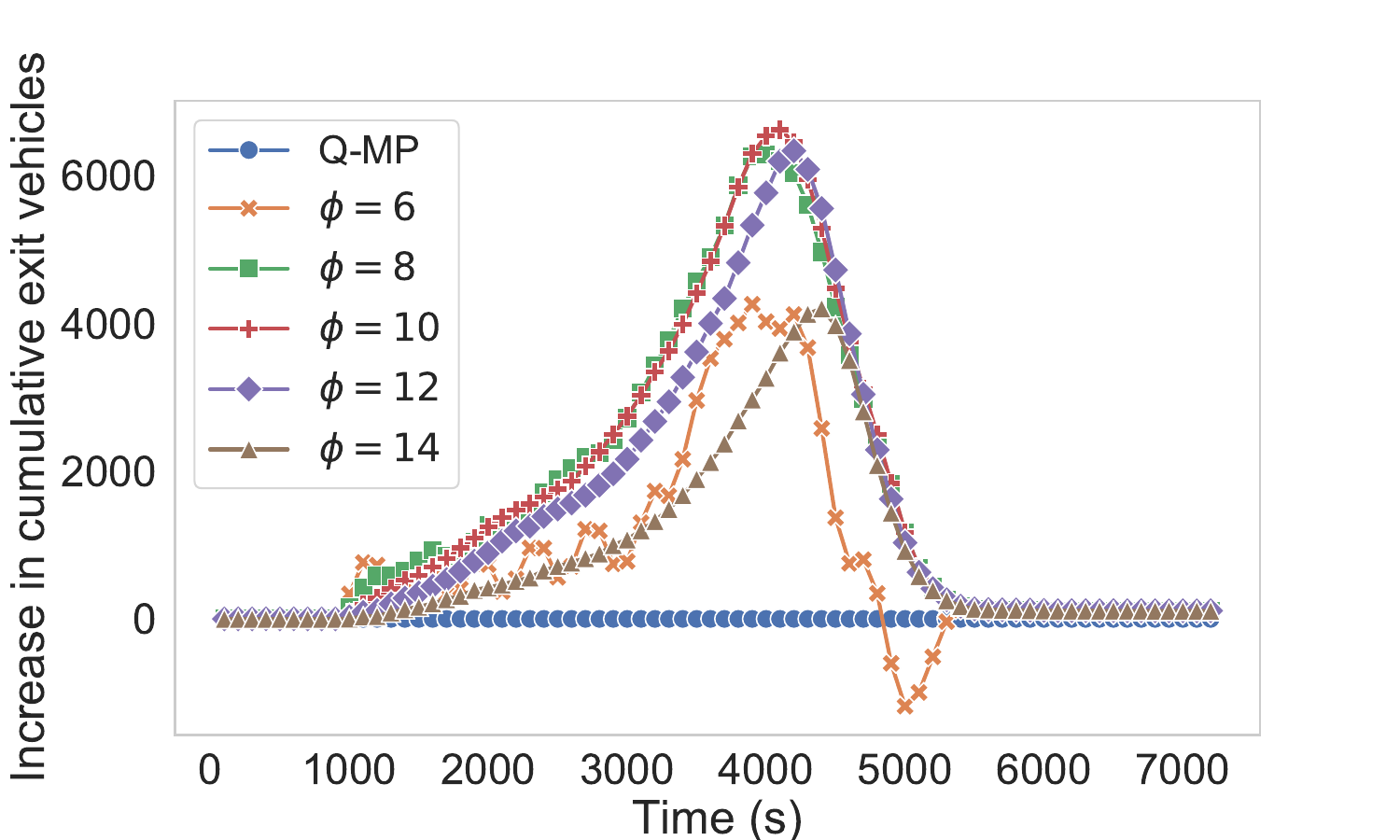}
		\caption{Comparison of exit vehicles}
		\label{fig:td on mfd}
	\end{subfigure}
	\begin{subfigure}[h]{0.49\textwidth}
		\centering
		\includegraphics[width=\textwidth]{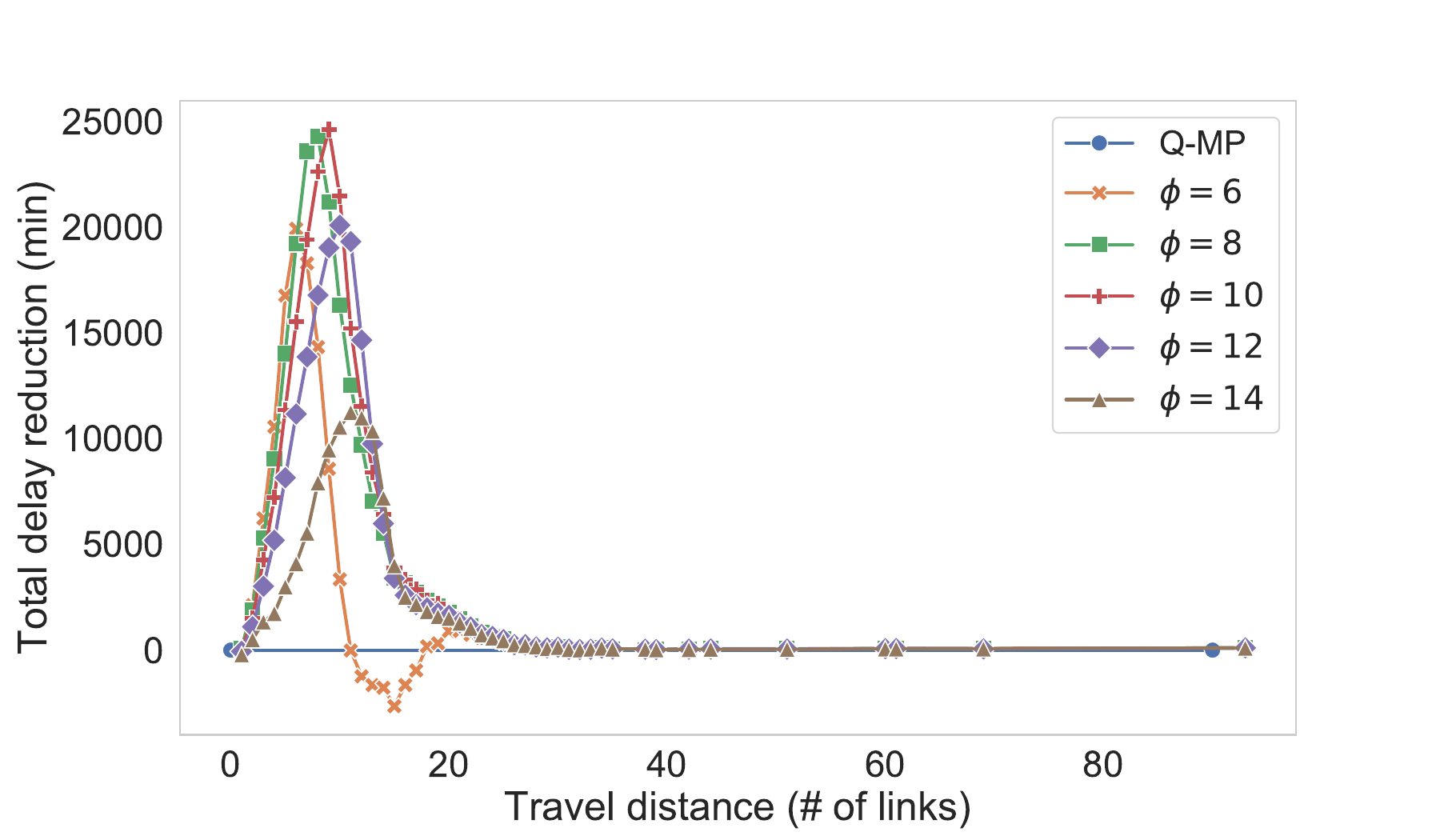}
		\caption{Vehicle delay}
		\label{fig:td on delay}
	\end{subfigure}
	\hfill
	\caption{Influence of $\phi$.}
	\label{fig:impact of td}
\end{figure}

The previous results do not impose a threshold on the holding time \footnote{The holding time is set to 5 mins for one-time holding. However, since a CAV can be held multiple times, the total holding time is not limited.}, which can lead to potential equity issues. To address this issue, we assume that all CAVs can be held at most once and vary the amount of time a vehicle is held, $\tau$. The results, shown in Figure \ref{fig:impact of tau}, reveal that for the tested range, a longer holding time leads to a higher operational efficiency. However, when $\tau\ge 8$ mins, the marginal improvement is less significant than $\tau< 8$ mins. Although $\tau= 10$ mins generates the best control performance, in practice, the selection for the value of $\tau$ should be determined based on the policy makers' evaluation on both the improvement for the system's efficiency and the sacrifice of individual's time from holding. To avoid holding vehicles for too long, $\tau=6$ mins is used, and all CAVs can be held at most once in the rest of this paper.

\begin{figure}[!ht]
	\centering
	\begin{subfigure}[h]{0.49\textwidth}
		\centering
		\includegraphics[width=\textwidth]{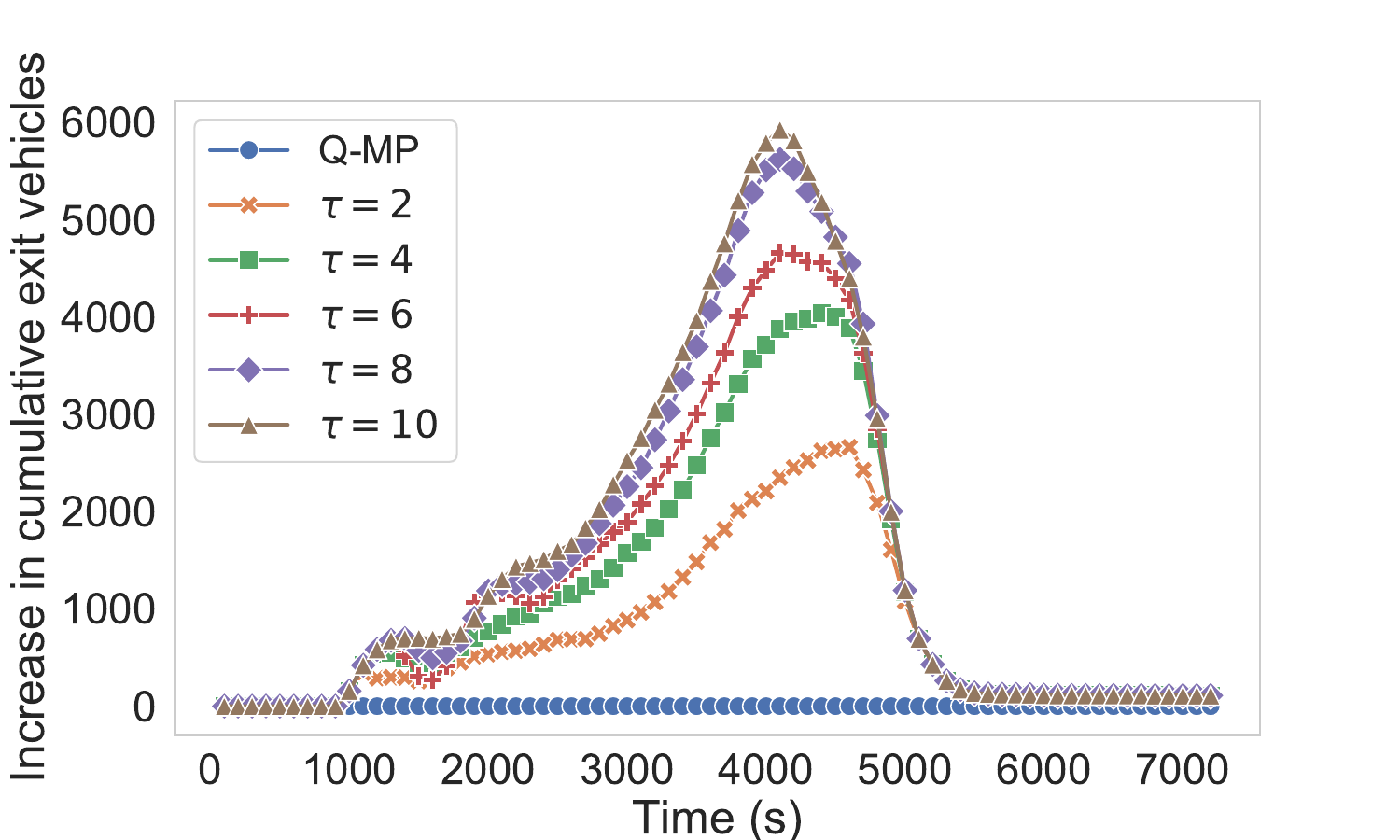}
		\caption{Comparison of exit vehicles}
		\label{fig:tau on mfd}
	\end{subfigure}
	\begin{subfigure}[h]{0.49\textwidth}
		\centering
		\includegraphics[width=\textwidth]{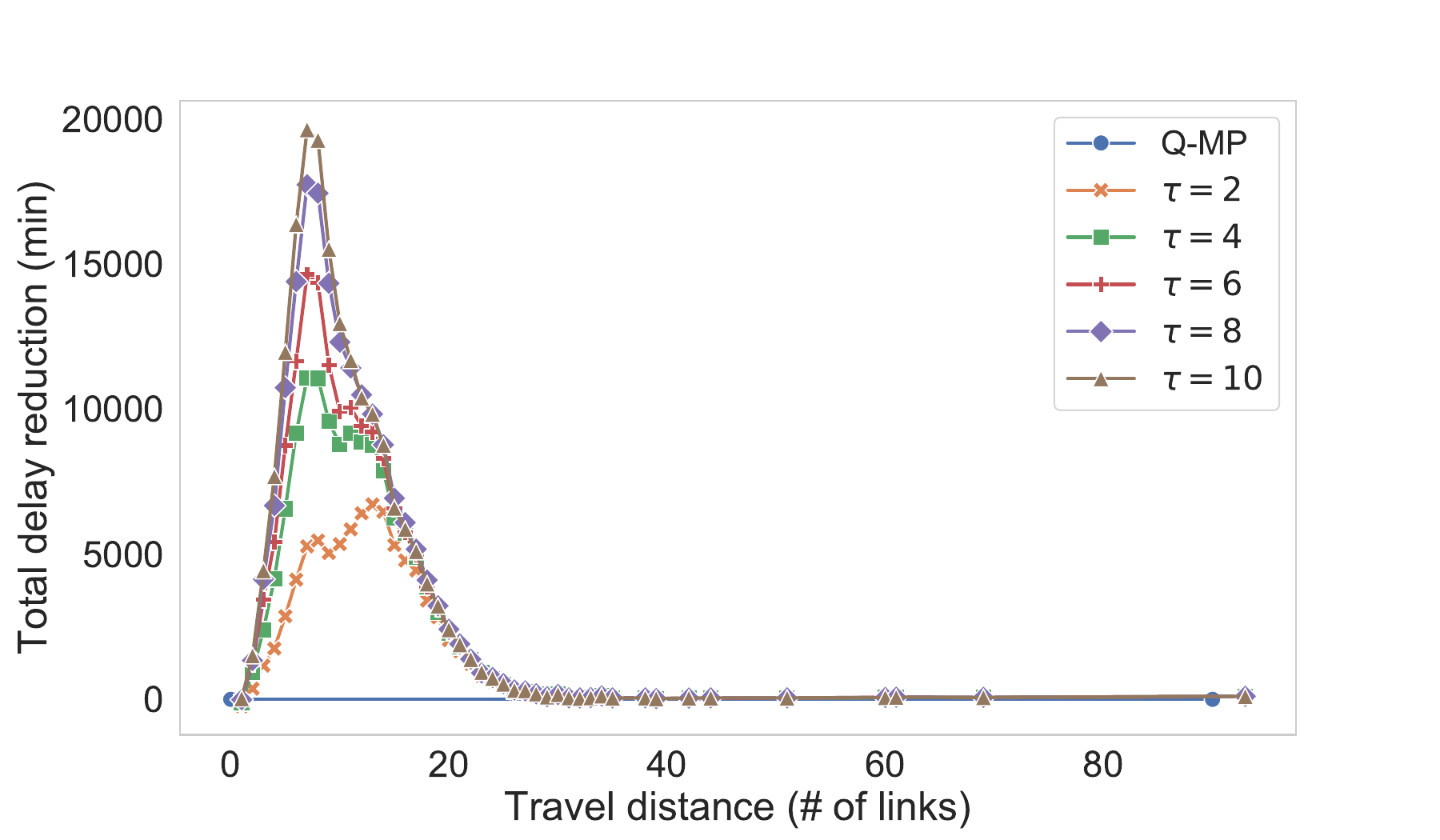}
		\caption{Vehicle delay}
		\label{fig:tau on delay}
	\end{subfigure}
	\hfill
	\caption{Influence of $\tau$ under $\phi=8$ links.}
	\label{fig:impact of tau}
\end{figure}

\subsubsection{Vehicle travel delay}
This section delves deeper into the effectiveness of the proposed algorithm on the reduction of travel delay for both held and non-held vehicles. For brevity, we only focus on the results for the combination of $\phi=8$ links and $\tau=6$ mins. However, the general trends and insights hold for other combinations of parameters.

Figure \ref{fig:total delay reduction both groups} shows the total delay reduction for both held and non-held vehicles associated with different travel distance. Note that even though the threshold for remaining travel distance is $\phi = 8$ links, meaning that vehicles with travel distances less than 8 links will never be held, there still exists held vehicleswith total travel distances less than 8 blocks in Figure \ref{fig:total delay reduction both groups}. This occurs because these travel distances indicated on the horizontal axis of Figure \ref{fig:total delay reduction both groups} represent the travel distance under the Q-MP (i.e., without holding); due to randomness and the dynamic routing used in the simulation, these vehicles might sometime take a route longer than $\phi$ blocks when holding is applied. As a result, they can be held and their travel delay will be increased, indicating by the negative delay reductions in Figure \ref{fig:total delay reduction both groups}. However, the percentage of these vehicles is very small as to be negligible; see Figure \ref{fig:percentage}.

\begin{figure}[!ht]
	\centering
	\includegraphics[width=0.6\textwidth]{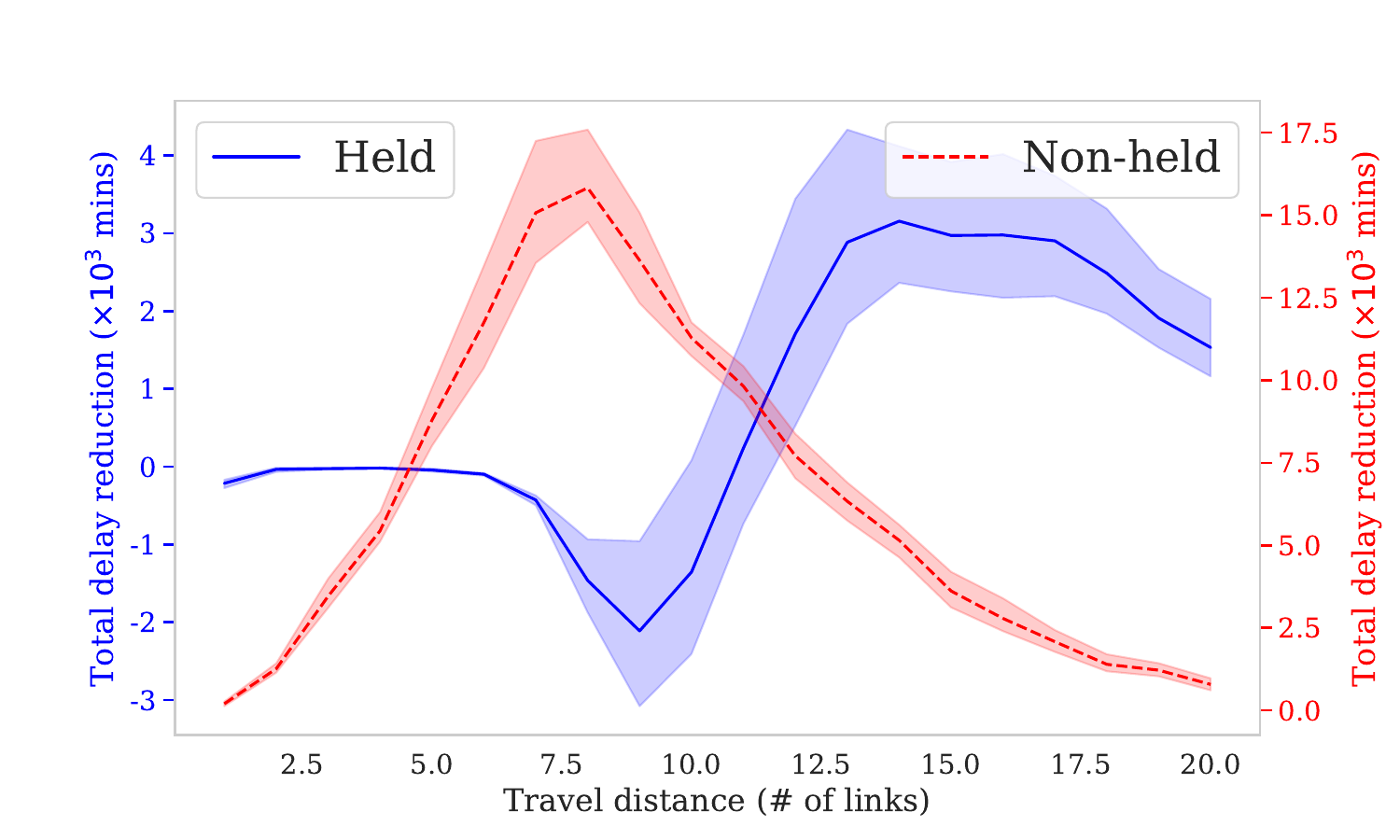}
	\caption{Total vehicle delay reduction.}
	\label{fig:total delay reduction both groups}
\end{figure}

\begin{figure}[h]
	\centering
	\includegraphics[width=0.6\textwidth]{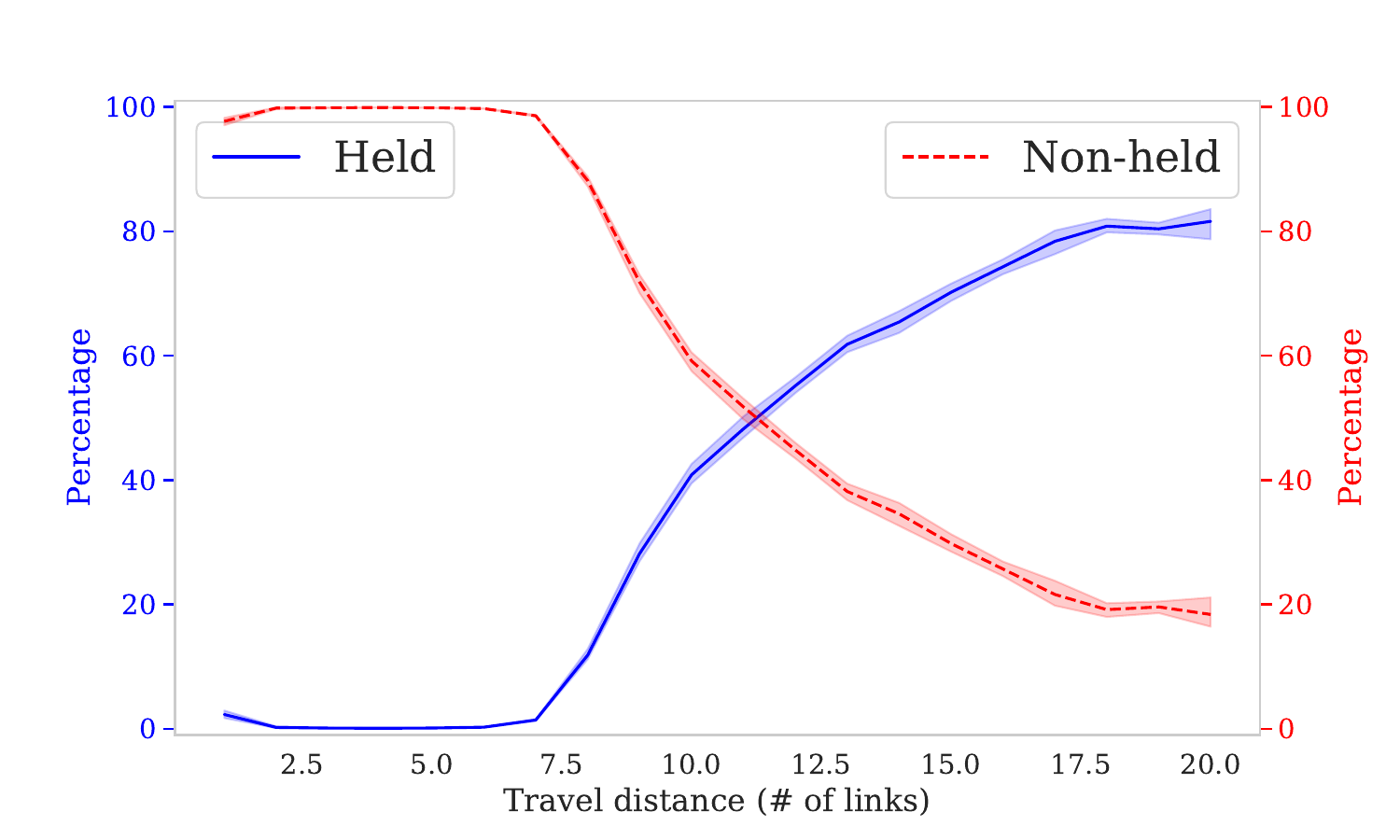}
	\caption{Percentage of vehicle groups.}
	\label{fig:percentage}
\end{figure}

In line with our expectation, travel delays are reduced for all vehicles that are not held; see the dashed red line in Figure \ref{fig:total delay reduction both groups}. Perhaps more importantly, however, we notice that although the travel delay of held-vehicles with a travel distance between 8-12 links increases, the travel delay for the held-vehicle groups with a longer travel distance actually decreases. This indicates that the holding strategy not only saves time for non-held vehicles, but it can also potentially benefit those CAVs that are held, leading to a reduction for total travel delay for all vehicle groups. This occurs because traffic conditions under the proposed strategy are improved compared to those that would arise without it, and longer trips tend to receive greater benefits from the reduced congestion. The observed delay reductions across these vehicle groups suggest that users in these groups can benefit from the long run, which represents a desirable property. This is an important feature of the proposed approach, as human drivers may also be more willing to participate if they recognize that temporary waiting can ultimately reduce their own travel times, thereby enhancing the practical implementation of the strategy. 

Figure \ref{fig:average delay reduction both groups} shows the average delay reduction per trip for both held and non-held vehicles. Interestingly, it indicates that for both held and non-held vehicles, the proposed control strategy brings greater advantages (or lower disadvantages for held vehicles that experience a delay increase) for vehicles with longer travel distances. It is reasonable to assume that, on average, the travel time reduction for a unit travel distance generated from the proposed strategy is identical for all vehicles. As a result, the total delay reduction increases with travel distance.  

\begin{figure}[h]
	\centering
	\includegraphics[width=0.6\textwidth]{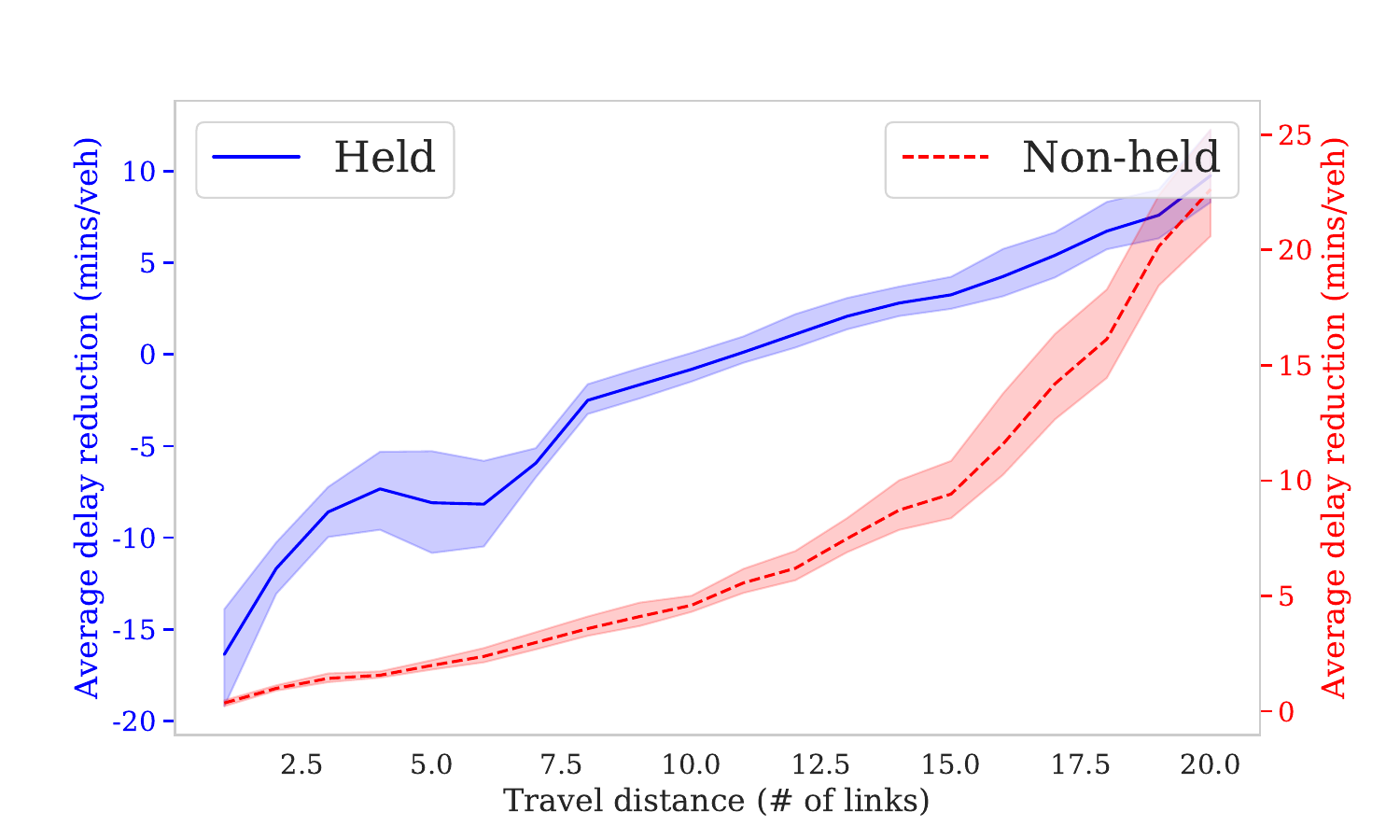}
	\caption{Average vehicle delay reduction.}
	\label{fig:average delay reduction both groups}
\end{figure}

Figure \ref{fig:percentage} shows the percentage of held and non-held vehicles associated with travel distance. Although Figure \ref{fig:average delay reduction both groups} shows that vehicles with longer travel distance can gain a larger average travel delay reduction, Figure \ref{fig:percentage} shows that the probability of holding a longer trip is higher. A vehicle will be held if two conditions are satisfied during its trip: 1) the average density exceeds the critical density; 2) the remaining travel distance is longer than $\phi$. The probability of both conditions being met increases for vehicles with longer travel distances, leading to a higher holding percentage. 

\subsubsection{Impact of delay from entering/exiting parking facilities}
Although parking facilities are associated with links such that CAVs do not need to cruise or reroute to search for parking, under congested conditions, additional time may still be required for vehicles to fully exit the network after initiating the parking maneuver. Therefore, we conduct an additional sensitivity analysis for the scenario with $\phi = 8$ links and $\tau = 6$ min, in which the enter/exit delay is modeled as a normally distributed random variable with mean $\mu_d$ ranging from 10 s to 60 s and a standard deviation of 5 s. The effective holding time is adjusted by subtracting this delay from the predefined holding duration $\tau = 6$ min. If a sampled delay is negative, it is set to zero; similarly, if the adjusted holding time becomes negative, it is truncated to zero. The impact of this delay on control performance is illustrated in Figure \ref{fig:impact of mu}, which shows that the strategy remains robust across the tested scenarios. As the enter/exit delay increases, the number of exiting vehicles and the total delay reduction decrease slightly, reflecting the reduced effective holding time. However, this impact appears to be negiglible on the overall operating performance.

\begin{figure}[!ht]
	\centering
	\begin{subfigure}[h]{0.49\textwidth}
		\centering
		\includegraphics[width=\textwidth]{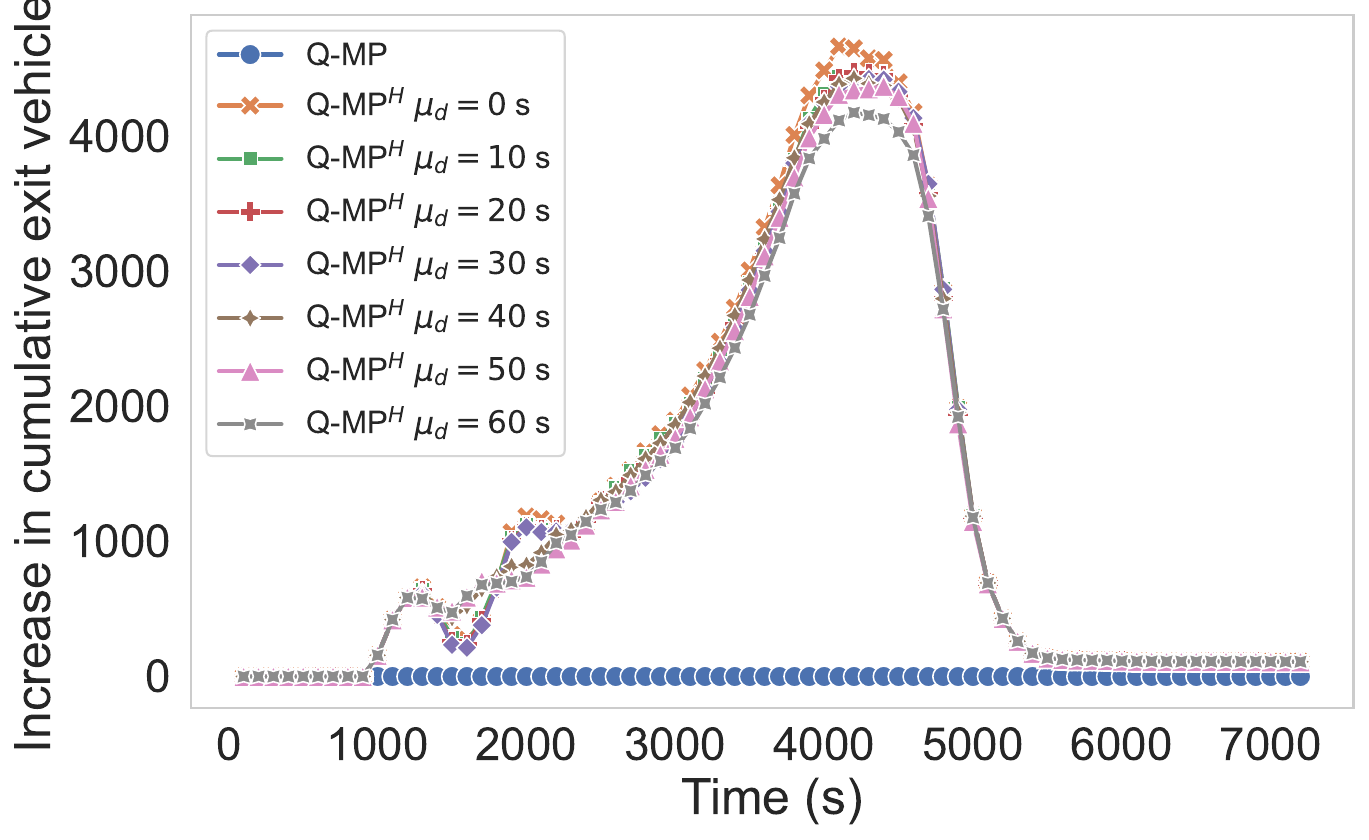}
		\caption{Comparison of exit vehicles}
		\label{fig:mu on mfd}
	\end{subfigure}
	\begin{subfigure}[h]{0.49\textwidth}
		\centering
		\includegraphics[width=\textwidth]{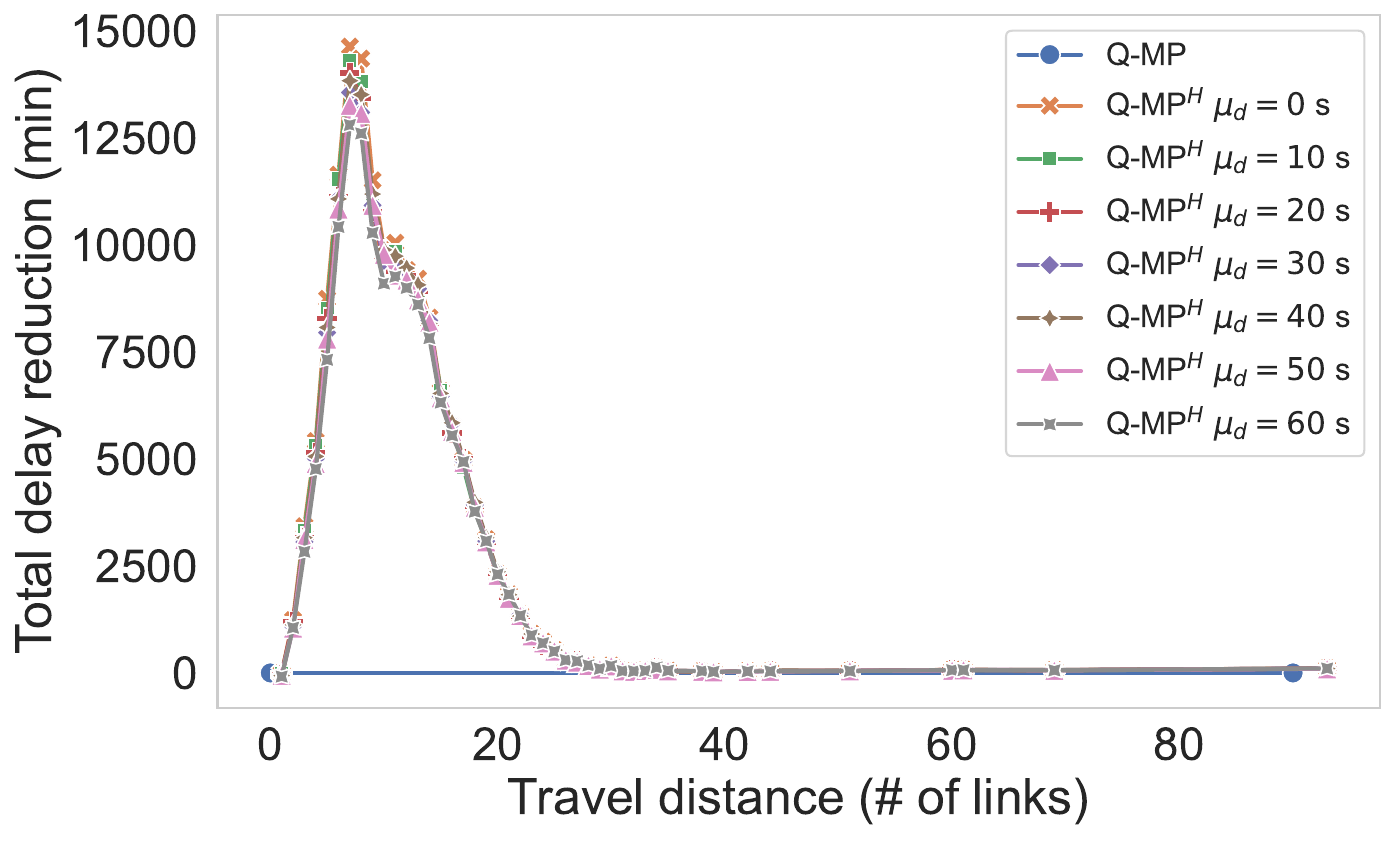}
		\caption{Vehicle delay}
		\label{fig:mu on delay}
	\end{subfigure}
	\hfill
	\caption{Influence of $\mu_d$ under $\phi=8$ links, $\tau=6$ mins.}
	\label{fig:impact of mu}
\end{figure}

\subsection{Results with restricted parking space}\label{sec: results restricted parking}
\subsubsection{Perimeter-like parking distribution with infinite parking capacity}\label{sec: perimeter parking dist}
The results in the previous section are obtained under the assumption that parking space is unrestricted, i.e., vehicles can leave the street immediately and wait at the parking location next to the street whenever needed. This assumption requires parking space to be available next to all streets, which is difficult to realize in practice. This section tests the control performance under the scenario in which parking spaces are restricted. We first consider that these parking spaces are located at similar locations to a perimeter boundary that would be implemented using perimeter metering control (i.e., a square boundary in this case). Later,  random parking space locations are considered that would be ``perimeter-less" (in Section \ref{sec:randomparking}). Only vehicles that are passing these locations when the holding condition is satisfied can be held. For the square boundary configuration, let index $i$ indicate the side length, in the unit of blocks, of the square. Five squares shown in Figure \ref{fig:parking} are selected. The purpose of this pattern is to make it comparable to the benchmark perimeter control algorithms shown in the next section.

\begin{figure}[!ht]
	\centering
	\includegraphics[width=0.5\textwidth]{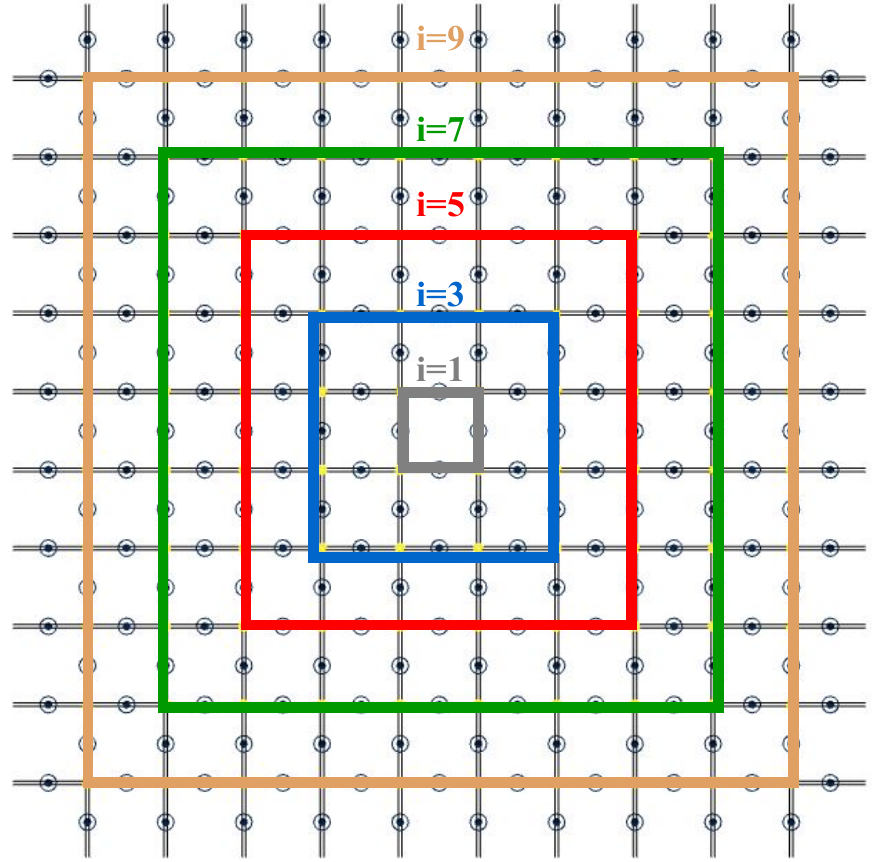}
	\caption{Locations of parking lots.}
	\label{fig:parking}
\end{figure}

Figure \ref{fig:impact of parking} shows the results associated with the five parking location patterns in Figure \ref{fig:parking}. Overall, $i=7$ generates the highest operational efficiency among the tested values. When $i=9$, although more parking locations are offered, the number of trips passing those locations is smaller than that when $i=7$. Consequently, fewer vehicles may be held, and the enhancement of the efficiency is less significant. On the other hand, a smaller value for $i$ can also lead to a reduction in the control performance for two reasons: 1) the number of parking locations is reduced; 2) the number of vehicles passing those locations that have a remaining travel distance longer than the threshold value is also reduced. For example, if we assume all vehicles are always using distance-based shortest path to their destination, when $i=1$, the longest remaining travel distance for vehicles passing any of the four centroids is 10 links. However, this maximum value for a centroid when $i=7$ is 16 links, which leads to a higher number of vehicles that can be potentially held. 

\begin{figure}[!ht]
	\centering
	\begin{subfigure}[h]{0.49\textwidth}
		\centering
		\includegraphics[width=\textwidth]{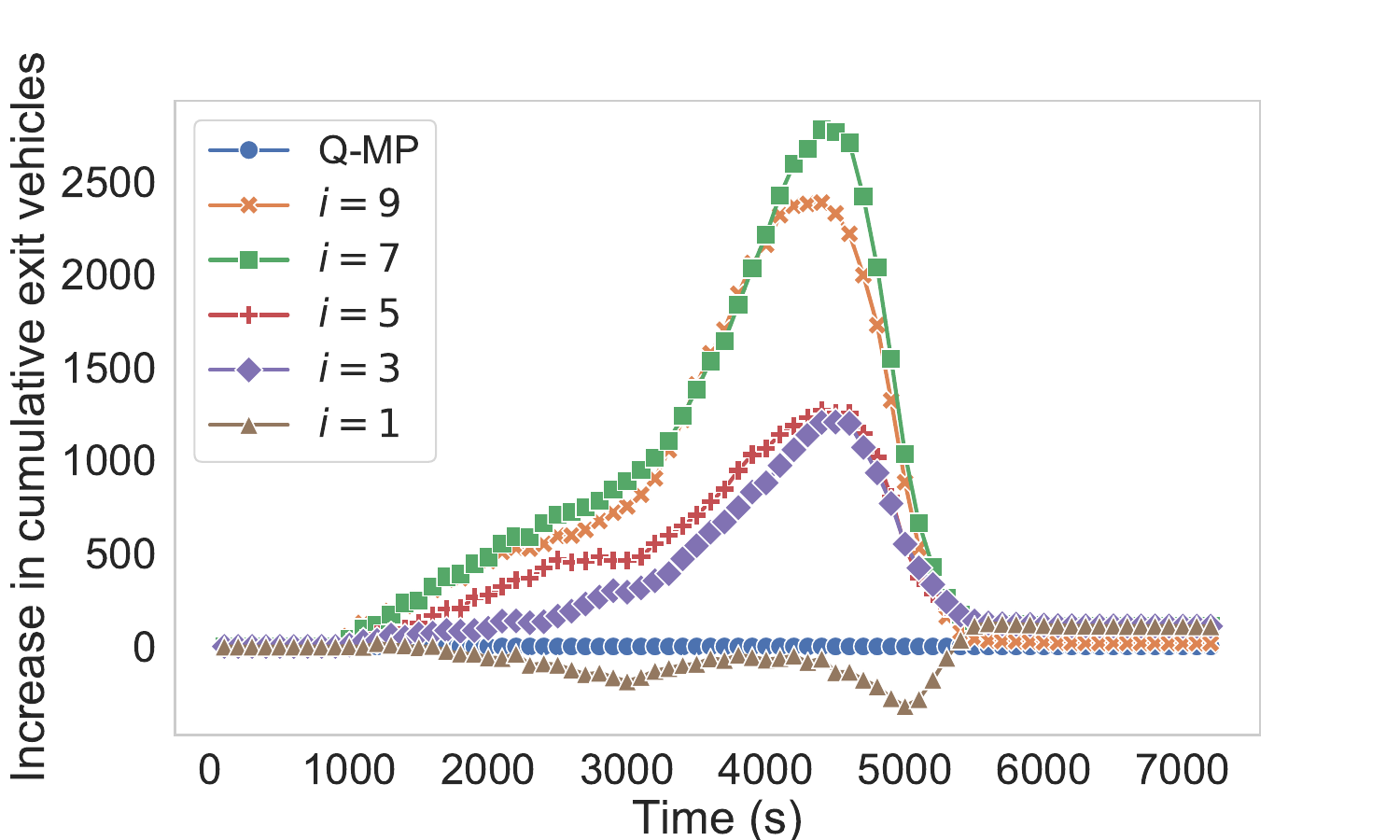}
		\caption{Comparison of exit vehicles}
		\label{fig:parking on mfd}
	\end{subfigure}
	\begin{subfigure}[h]{0.49\textwidth}
		\centering
		\includegraphics[width=\textwidth]{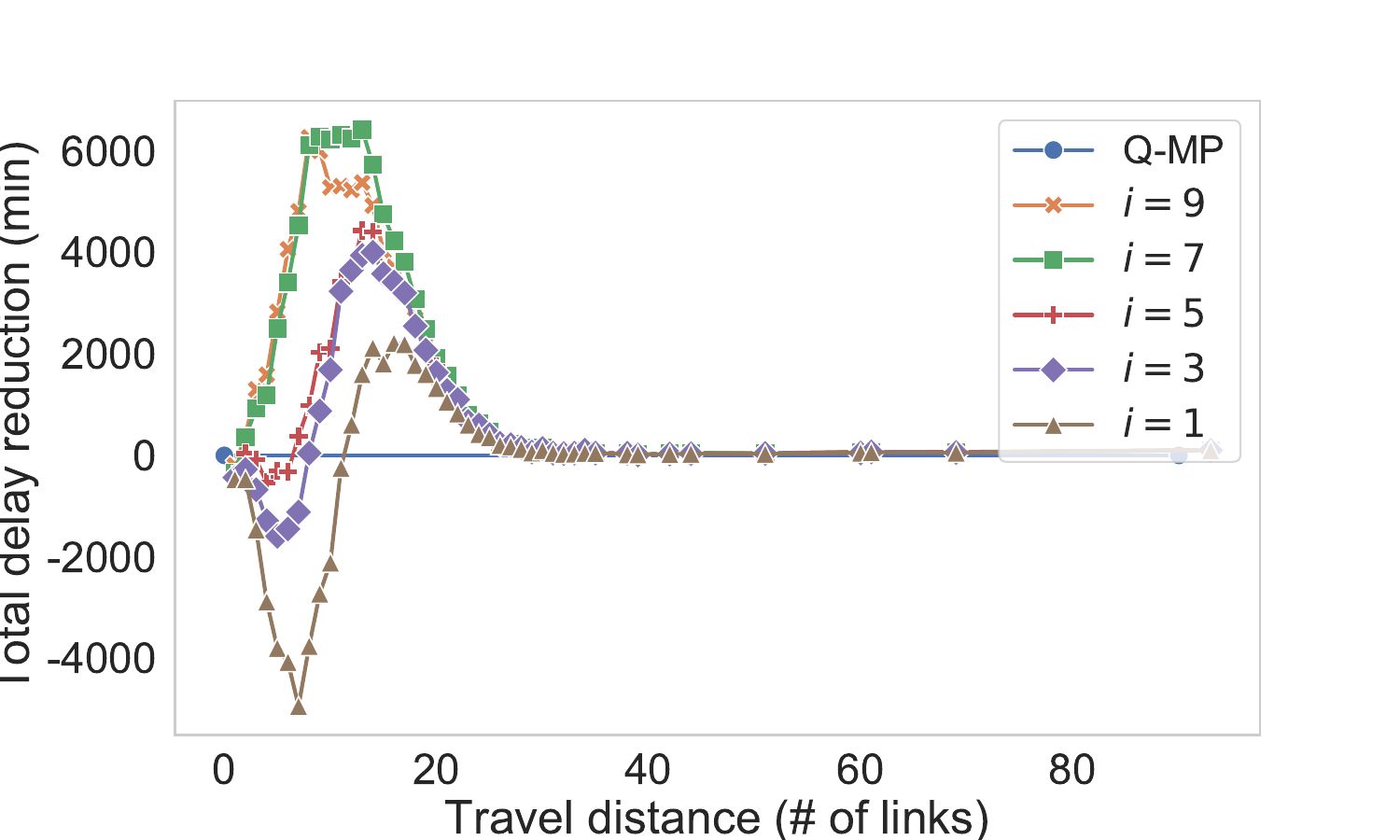}
		\caption{Vehicle delay}
		\label{fig:parking on delay}
	\end{subfigure}
	\hfill
	\caption{Influence of parking space.}
	\label{fig:impact of parking}
\end{figure}

The results in this section also inherently assume that sufficient parking space is available at each of the parking locations to accommodate vehicles that should be held. To verify if this is a realistic assumption, Figure \ref{fig:n of veh per cen} shows the minimum, average, and maximum numbers of holding vehicles across the 28 parking locations when $i=7$. Although the average number of vehicles at a parking location is generally below 25, which should not be difficult to accommodate by a typical parking garage or parking lots, the value can grow as large as 60, which exceeds the capacity for relatively small parking facilities. Thus, it is also imperative to consider the performance of the proposed strategy when parking capacity limits exist.

\begin{figure}[h]
	\centering
	\includegraphics[width=0.7\textwidth]{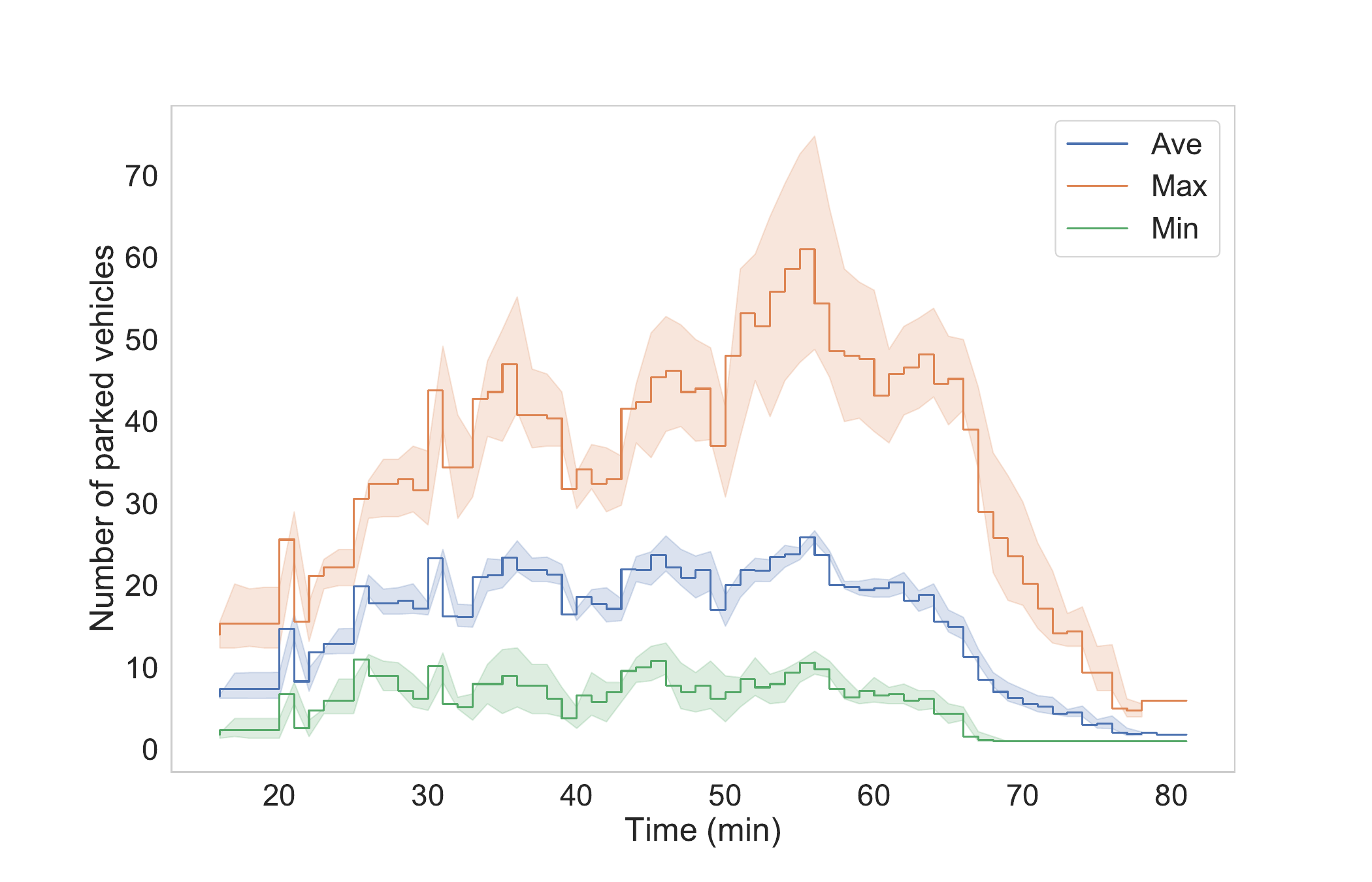}
	\caption{Average number of holding vehicle in each parking location when $i=7$.}
	\label{fig:n of veh per cen}
\end{figure}

\subsubsection{Perimeter-like parking distribution with  parking capacity limits}\label{sec: results with parking limit}

One method to overcome the parking capacity limitation is to reduce the number of vehicles that are required to park by increasing $\phi$ and/or decreasing $\tau$. However, doing so would reduce the effectiveness of the strategy, as shown in the previous sections. To investigate the performance under more realistic conditions, we restrict the number of parking spots at each parking location, $\bar{N}$, to one of the values in the set \{10, 20, 30, 40\}. When the number of held vehicles at a parking location reaches its capacity limit, no additional CAVs can be held there until some previously held CAVs have left the location. 

Figure \mbox{\ref{fig:influence of parking limit}} shows the impact of the parking capacity on the control performance. The results reveal that when $\bar{N}>=30$, the performance is essentially the same as when parking capacity is unrestricted. However, when $\bar{N}<30$, the performance of the proposed strategy is impacted by the parking capacity restriction. For this case, adding more parking spaces can accommodate more vehicles and, as a result, enhance the control performance. These results: 1) suggest existing parking capacity - no matter how small - can be utilized to improve performance; and, 2) provide insights on the evaluation of parking capacity increases and the corresponding benefits on operational performance that they can provide.

\begin{figure}[h]
	\centering
	\includegraphics[width=0.5\textwidth]{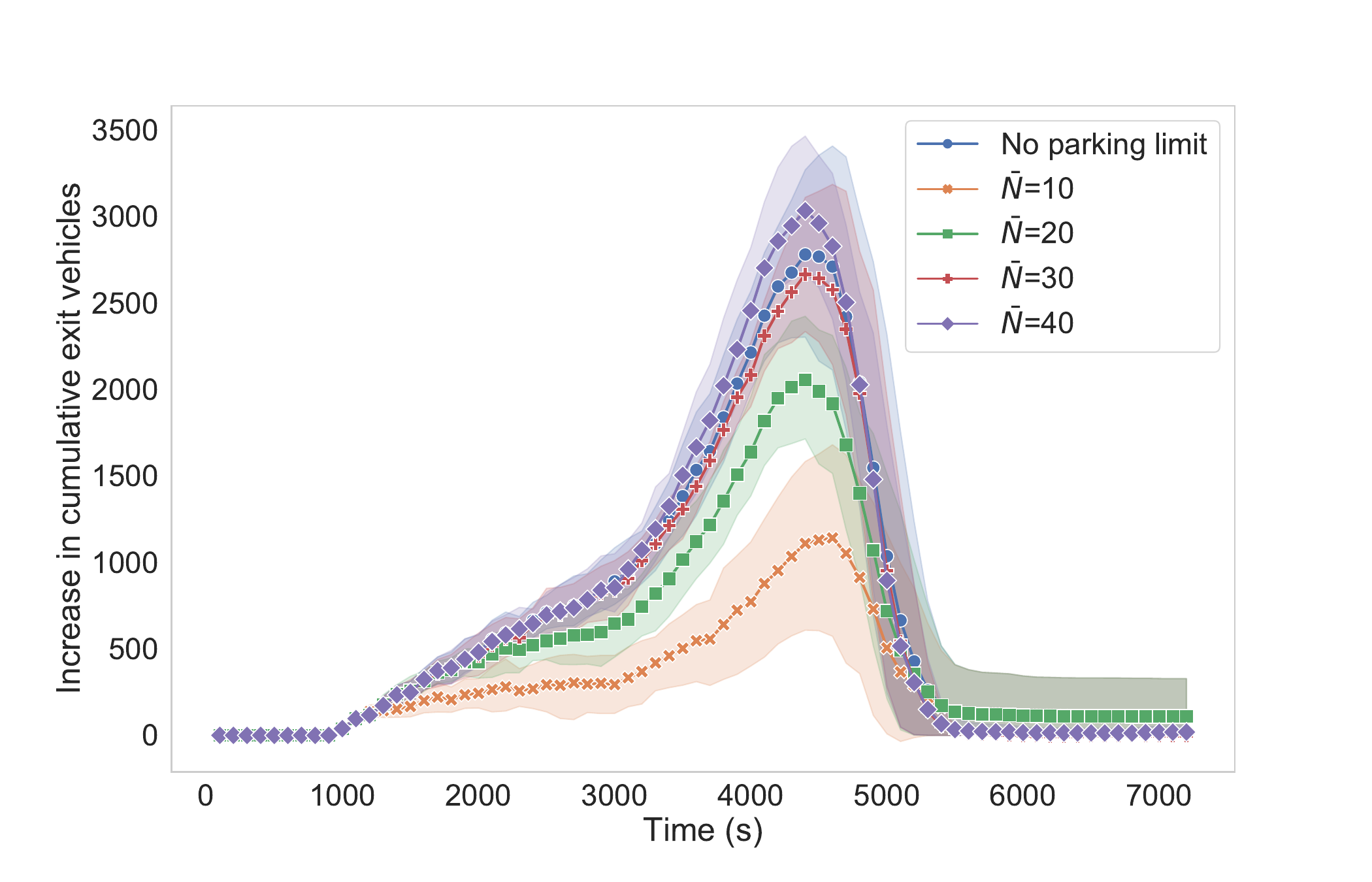}
	\caption{Influence of $\bar{N}$ on Q-MP\textsuperscript{H}.}
	\label{fig:influence of parking limit}
\end{figure}

Figure \mbox{\ref{fig:influence of parking limit}} also shows that $\bar{N}=40$ generates a slightly higher (statistically insignificant) vehicle exit rate than the case with unrestricted parking space. The reason is that, as shown in Figure \mbox{\ref{fig:compare of number of hold}}, when $\bar{N}=40$, the average and minimum number of holding vehicles across the 28 parking locations are similar to the case without parking restriction, while the maximum number of holding vehicles would be significantly reduced. This finding unveils that, on average, a parking space restriction of $\bar{N}=40$ is enough to accommodate vehicles to achieve an efficiency close to the situation as if there was no parking restriction. In the meanwhile, unrestricted parking scenario can not only generate unnecessarily too many vehicles at certain parking locations but also reduce the efficiency of holding vehicles returning to the network, which can be avoided by adding appropriate parking restriction.  

\begin{figure}[h]
	\centering
	\includegraphics[width=0.7\textwidth]{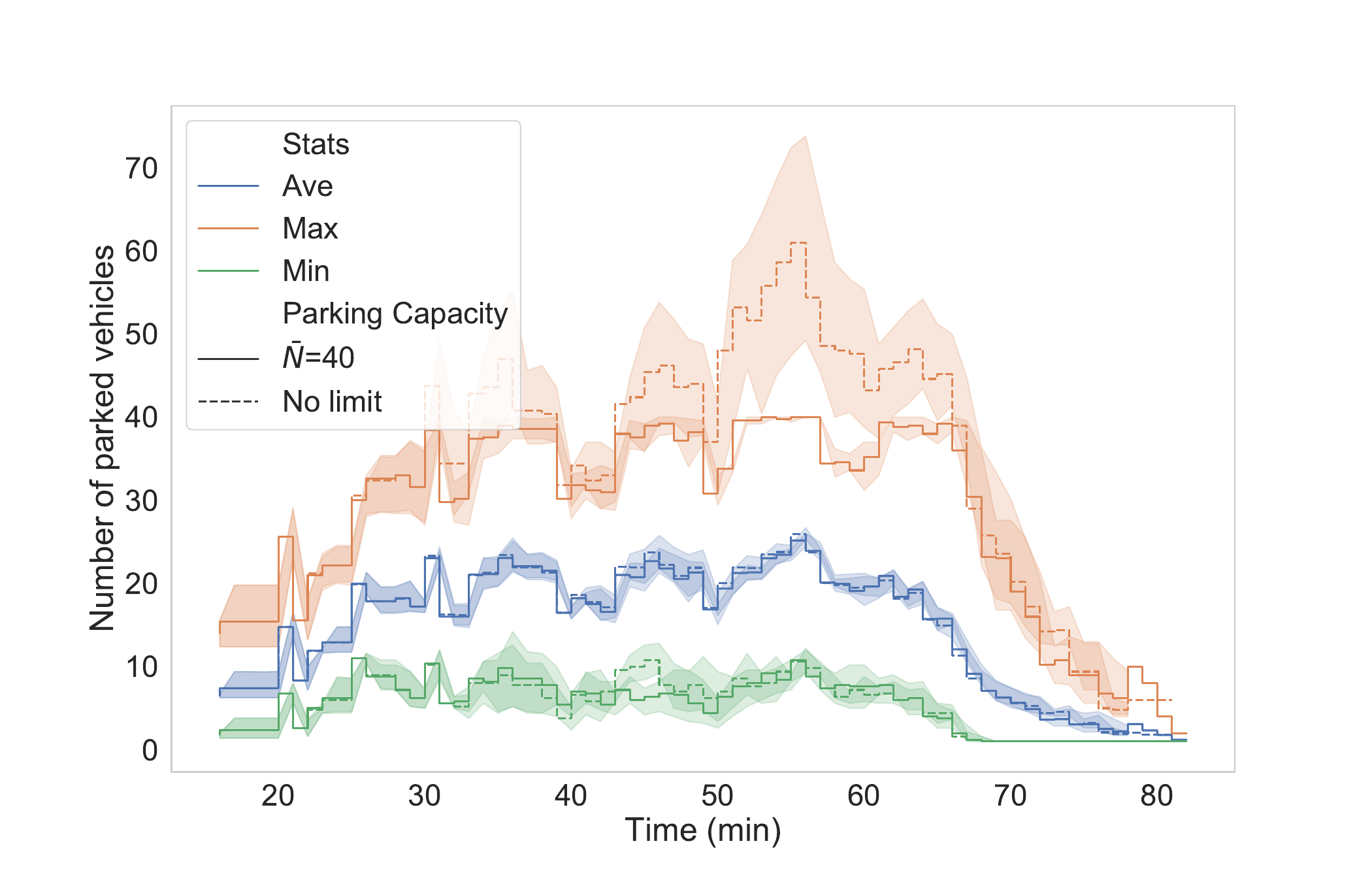}
	\caption{Comparison of number of holding vehicles.}
	\label{fig:compare of number of hold}
\end{figure}

\subsubsection{Comparison with two perimeter control algorithms}\label{sec:comparison_perimeter}

The pattern of parking locations shown in Figure \ref{fig:parking} is similar to the perimeters used in many perimeter control methods. Therefore, it is worth comparing the proposed holding strategy, which makes use of the available infrastructure to temporarily mitigate congestion, to perimeter control methods, which block inflows to temporarily reduce traffic burden for the congested region. Two perimeter control algorithms -- Bang-Bang control and N-MP -- are used for the comparison.

For simplicity, we only consider the parking pattern associated with $i=7$ in Figure \ref{fig:parking} and set perimeter intersections for both perimeter control algorithms to be on the square boundary. Under the uniform demand setting depicted in Section \ref{sec: net setup}, Bang-Bang control generates worse performance than the pure Q-MP control. This is because under the uniform demand pattern, there exist significant outflows from the protected region. After implementing Bang-Bang control, the queue accumulation at the external side of the perimeter can further block outflows from the protected region, which reduces the operational efficiency considerably. For simplicity, the results under this condition are not shown. To ensure a fair comparison, the concentrated demand pattern defined in Section \ref{sec: net setup} is utilized to avoid this blocking effect. 

As introduced in Section \ref{sec: perimeter base}, the critical density $\rho^p_{cr}$ is the condition for the activation of inflow restriction for perimeter control algorithms. Therefore, using Bang-Bang control as the baseline algorithm, we tested the following values \{30, 35, 40, 45, 50\} veh/(lane-km) for $\rho^p_{cr}$. Figure \ref{fig:bangbang} shows that Bang-Bang control with an appropriate critical density can improve the control performance effectively. Accordingly, $\rho^p_{cr}=40$ veh/(lane-km), which generates the best performance for Bang-Bang control, is selected as the critical density for N-MP and Q-MP\textsuperscript{H} as well. Additionally, $\xi$ in Eq. \eqref{eq:psi} is manually tuned for N-MP. Ten equally spaced values between 0.2 to 2 are tested for $\xi$, and the value of $\xi=1.4$ leads to the best performance in terms of the network exit rate. The results are omitted for simplicity.

\begin{figure}[h]
	\centering
	\includegraphics[width=0.5\textwidth]{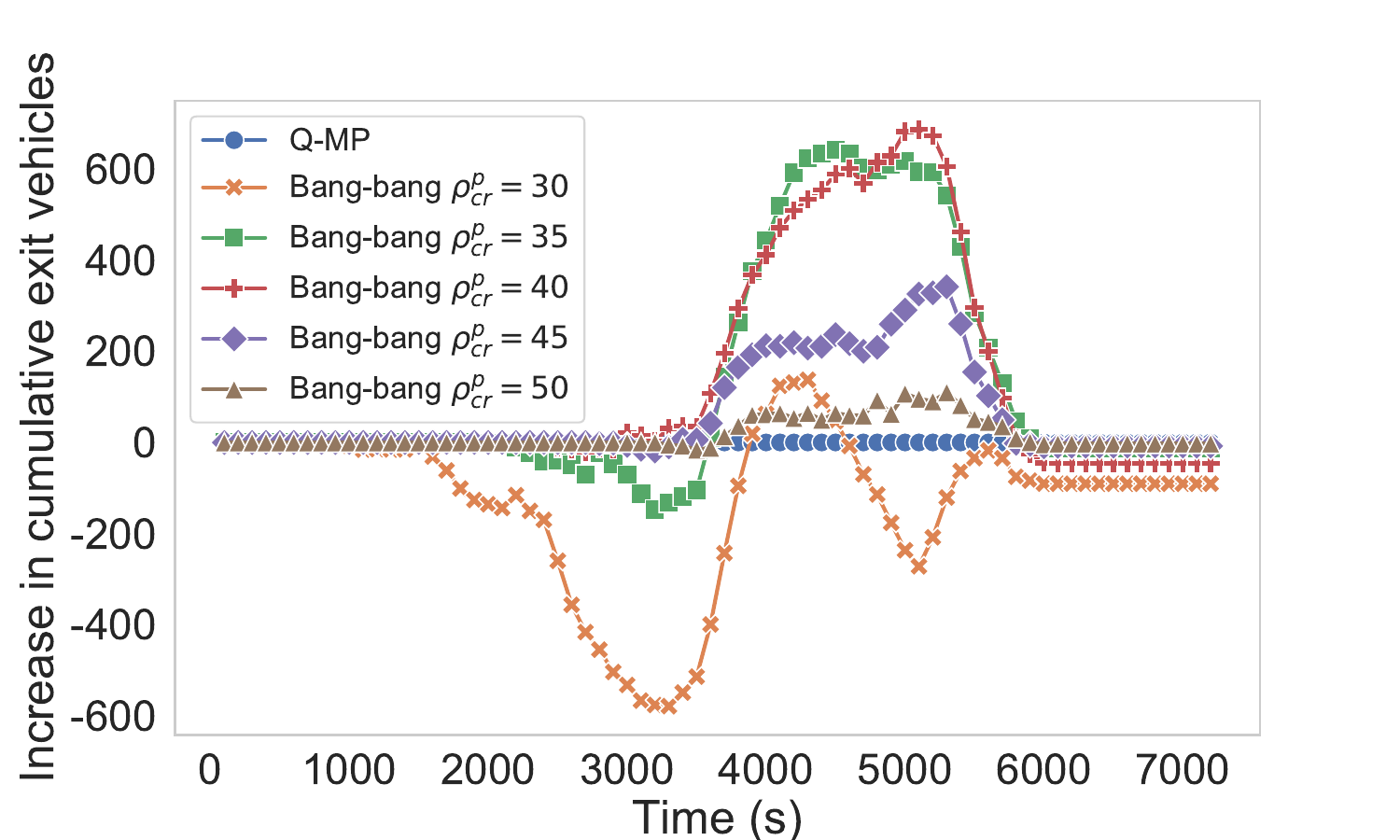}
	\caption{Influence of $\rho^p_{cr}$ on Bang-Bang control.}
	\label{fig:bangbang}
\end{figure}

The comparison of cumulative number of exit vehicles between Q-MP\textsuperscript{H}, Bang-Bang control, and N-MP is shown in 
Figure \ref{fig:comparison}, in which Bang-Bang control is used as the baseline. Note that due to the removal of outgoing trips from the protected region, the average remaining travel distance passing the centroids on the perimeter is much shorter. Consequently, the threshold value for the remaining travel distance defined in Section \ref{sec:phi and tau}, i.e., $\phi = 8$, is too large to hold enough vehicles to promote significant improvement. Therefore, we select a smaller value, $\phi=5$ links, for this simulation.  Both N-MP and Q-MP\textsuperscript{H} outperform Bang-Bang control. Interestingly, N-MP generates a similar peak value for the increase in the cumulative number of exit vehicles to Q-MP\textsuperscript{H}; however, Q-MP\textsuperscript{H} obtains its higher performance more quickly, which leads to further reduction in travel time. It needs to be emphasized that since all three algorithms utilize the same critical density value, this advantage does not result from an earlier activation of the control strategy. Instead, it is because the queue accumulation at the perimeter from Q-MP\textsuperscript{H} is shorter than N-MP, which mitigates the negative impact on the overall efficiency during the congestion formation process. N-MP catches Q-MP\textsuperscript{H} after the traffic congestion diminishes during the cool down period. 

\begin{figure}[h]
	\centering
	\includegraphics[width=0.5\textwidth]{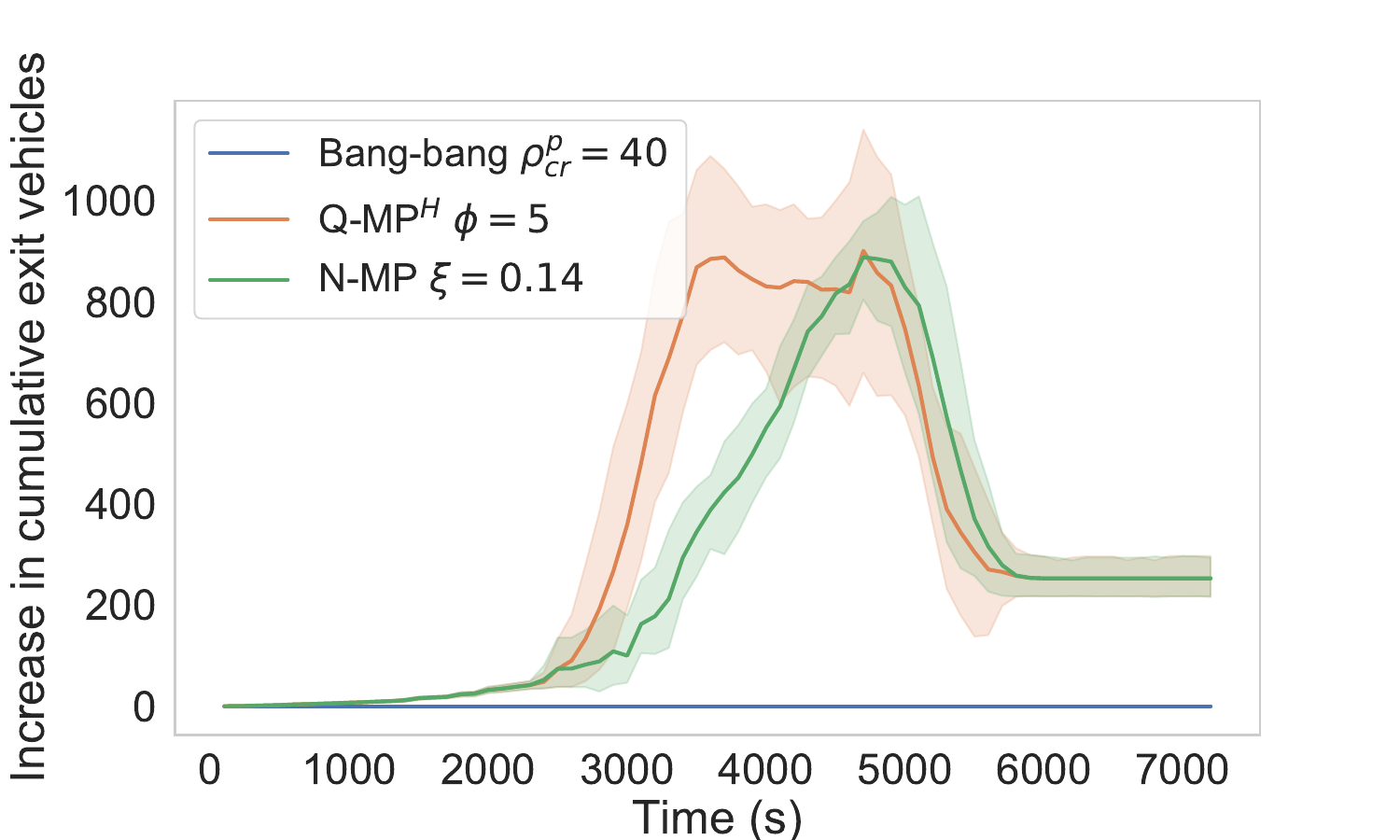}
	\caption{Comparison between Q-MP\textsuperscript{H}, Bang-Bang control, and N-MP.}
	\label{fig:comparison}
\end{figure}

\subsubsection{Random parking locations with infinite parking capacity}\label{sec:randomparking}

All tested patterns for parking locations in Figure \ref{fig:parking} form regular squares around the center of the network; however, such configurations may not be available for most networks. For each configuration shown in Figure \ref{fig:parking}, we randomly select the same number of parking locations from the entire network to apply the same holding strategy. Doing so demonstrates the flexibility of the proposed method and the ``perimeter-less" nature of the metering that it can achieve. $i=1$ is excluded since it worsen the traffic operation compared to Q-MP, as shown in Figure \ref{fig:impact of parking}. For each $i$, 10 sub-configurations with different random seeds are generated. Figure \ref{fig:example of random parking locations} shows examples of randomly generated parking locations for the four tested indices.

\begin{figure}[!ht]
	\centering
	\begin{subfigure}[h]{0.49\textwidth}
		\centering
		\includegraphics[width=\textwidth]{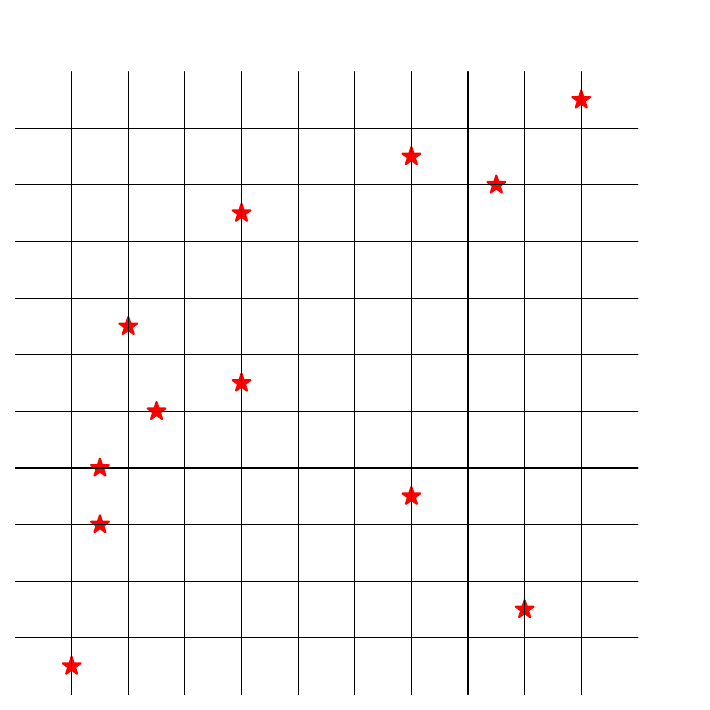}
		\caption{$i=3$}
		\label{fig:i=3 loc}
	\end{subfigure}
	\begin{subfigure}[h]{0.49\textwidth}
		\centering
		\includegraphics[width=\textwidth]{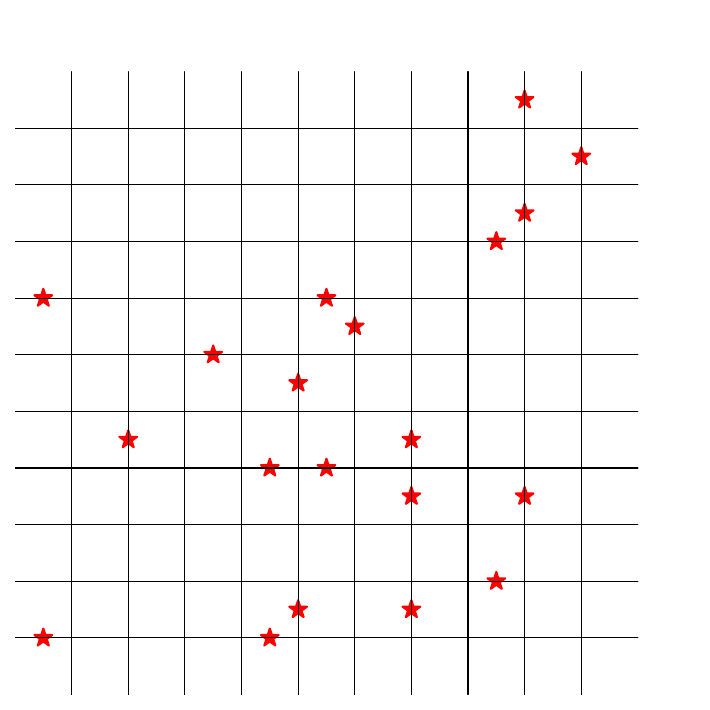}
		\caption{$i=5$}
		\label{fig:i=5 loc}
	\end{subfigure}
	
	\centering
	\begin{subfigure}[h]{0.49\textwidth}
		\centering
		\includegraphics[width=\textwidth]{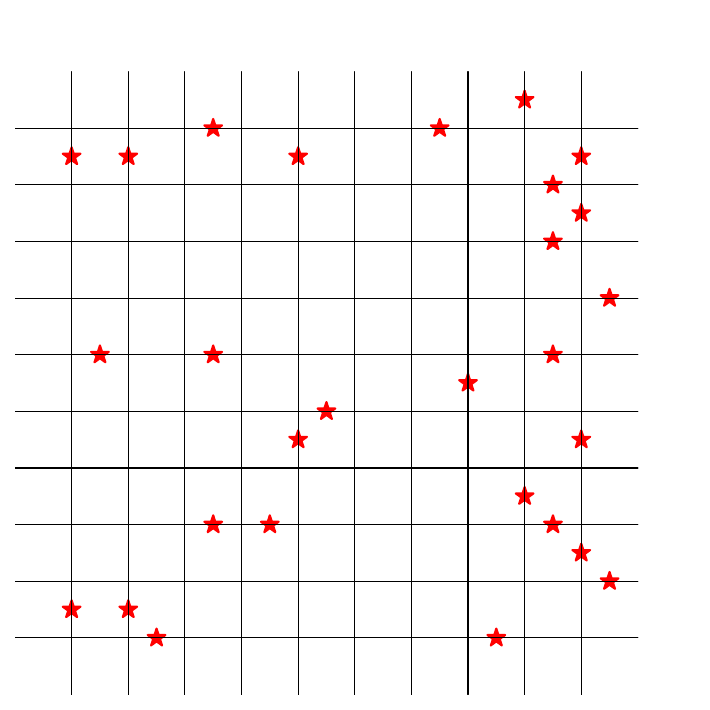}
		\caption{$i=7$}
		\label{fig:i=7 loc}
	\end{subfigure}
	\begin{subfigure}[h]{0.49\textwidth}
		\centering
		\includegraphics[width=\textwidth]{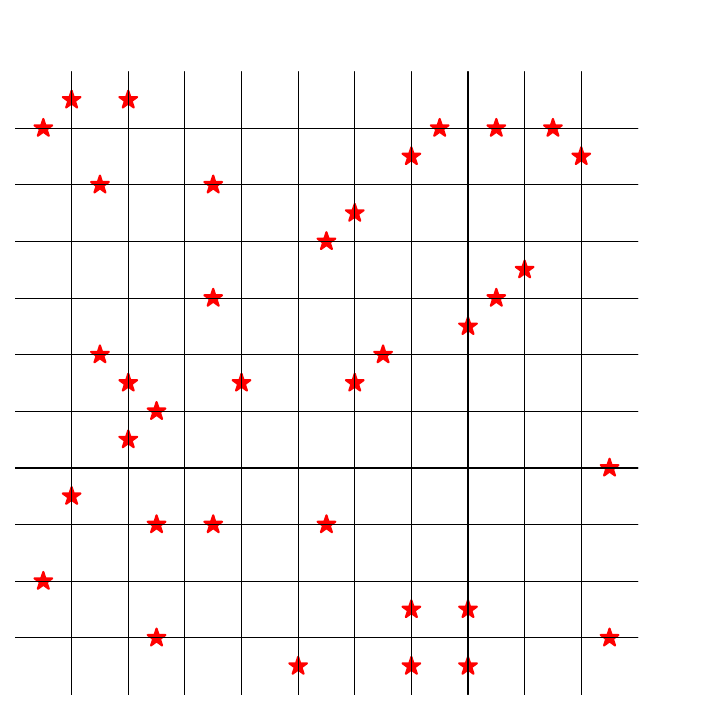}
		\caption{$i=9$}
		\label{fig:i=9 loc}
	\end{subfigure}
	\hfill
	\caption{Examples of random parking locations.}
	\label{fig:example of random parking locations}
\end{figure}

The simulation results are shown in Figure \ref{fig:impact of random parking location}. It shows that except for $i=7$, random parking locations can help foster the control performance, which is in line with our expectation. As mentioned before, when $i=9$, all parking locations are distributed at the boundary of the network, and the number of vehicles passing those areas is limited. As shown in Figure \ref{fig:i=9 loc}, adding randomness can help overcome this drawback and make a better use of those parking locations. Similarly, when $i=5$, the issue of short remaining travel distance existing in regular parking locations can also be mitigated by the existence of parking locations farther away from the center from random distributions, as shown by Figure \ref{fig:i=5 loc}. When $i$ is very small, i.e., $i=3$, there are not enough parking locations to hold vehicles to generate significant improvement from random parking locations. Moreover, Figure \ref{fig:i=7} shows that the random parking distribution does not have significant impact on the performance when $i=7$, which implies that the regular parking location pattern is already appropriate to satisfy the need of the holding strategy to produce a favorable improvement.

\begin{figure}[!ht]
	\centering
	\begin{subfigure}[h]{0.49\textwidth}
		\centering
		\includegraphics[width=\textwidth]{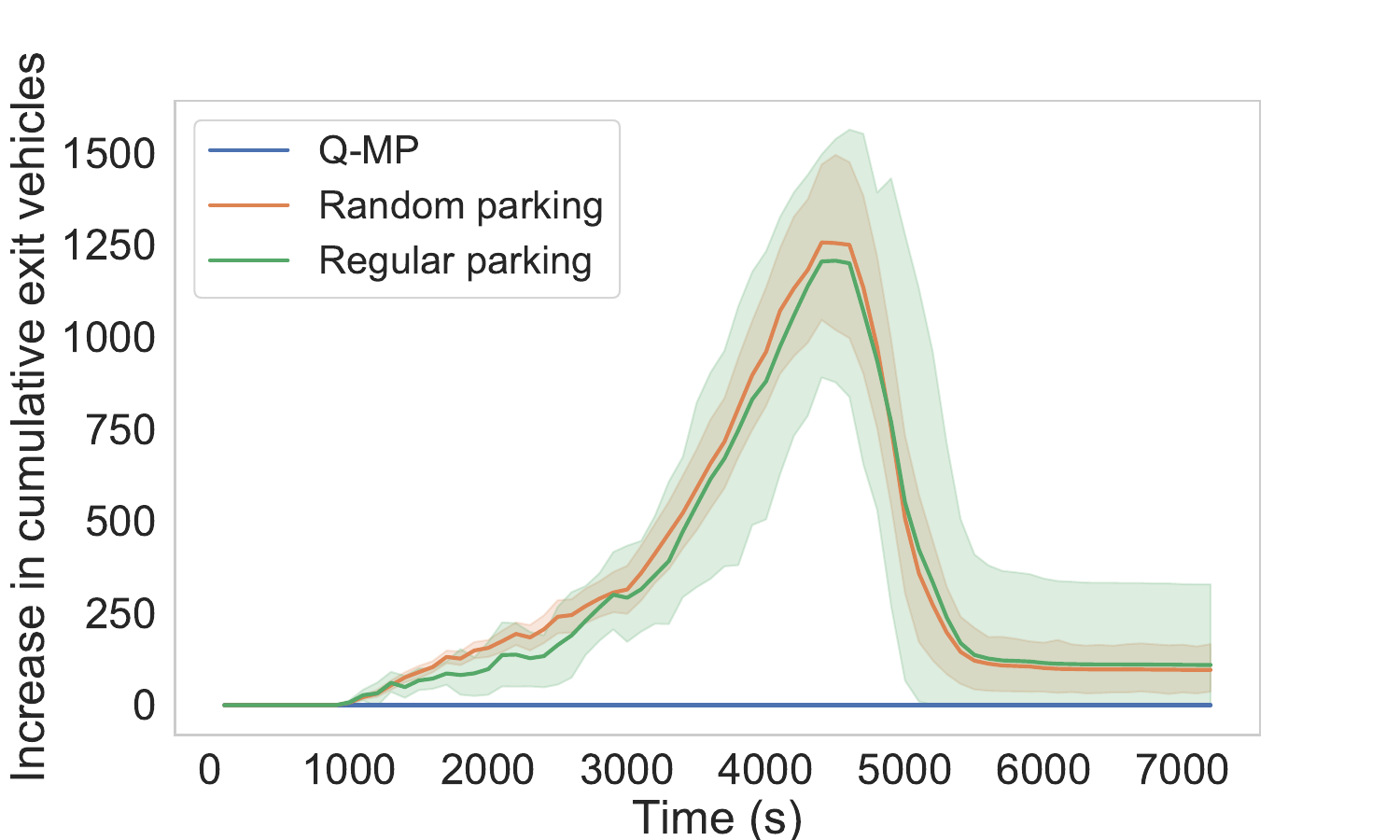}
		\caption{$i=3$}
		\label{fig:i=3}
	\end{subfigure}
	\begin{subfigure}[h]{0.49\textwidth}
		\centering
		\includegraphics[width=\textwidth]{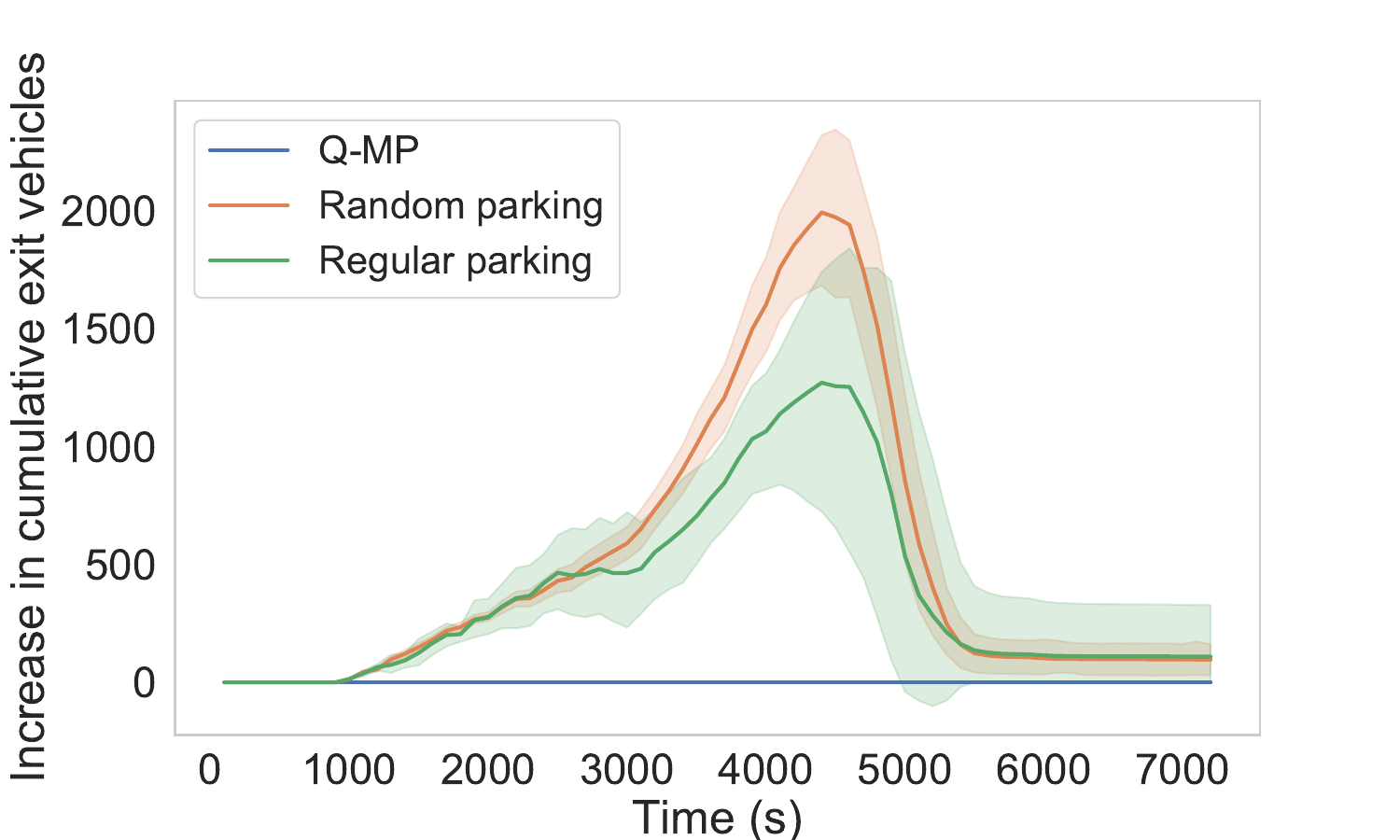}
		\caption{$i=5$}
		\label{fig:i=5}
	\end{subfigure}
	
	\centering
	\begin{subfigure}[h]{0.49\textwidth}
		\centering
		\includegraphics[width=\textwidth]{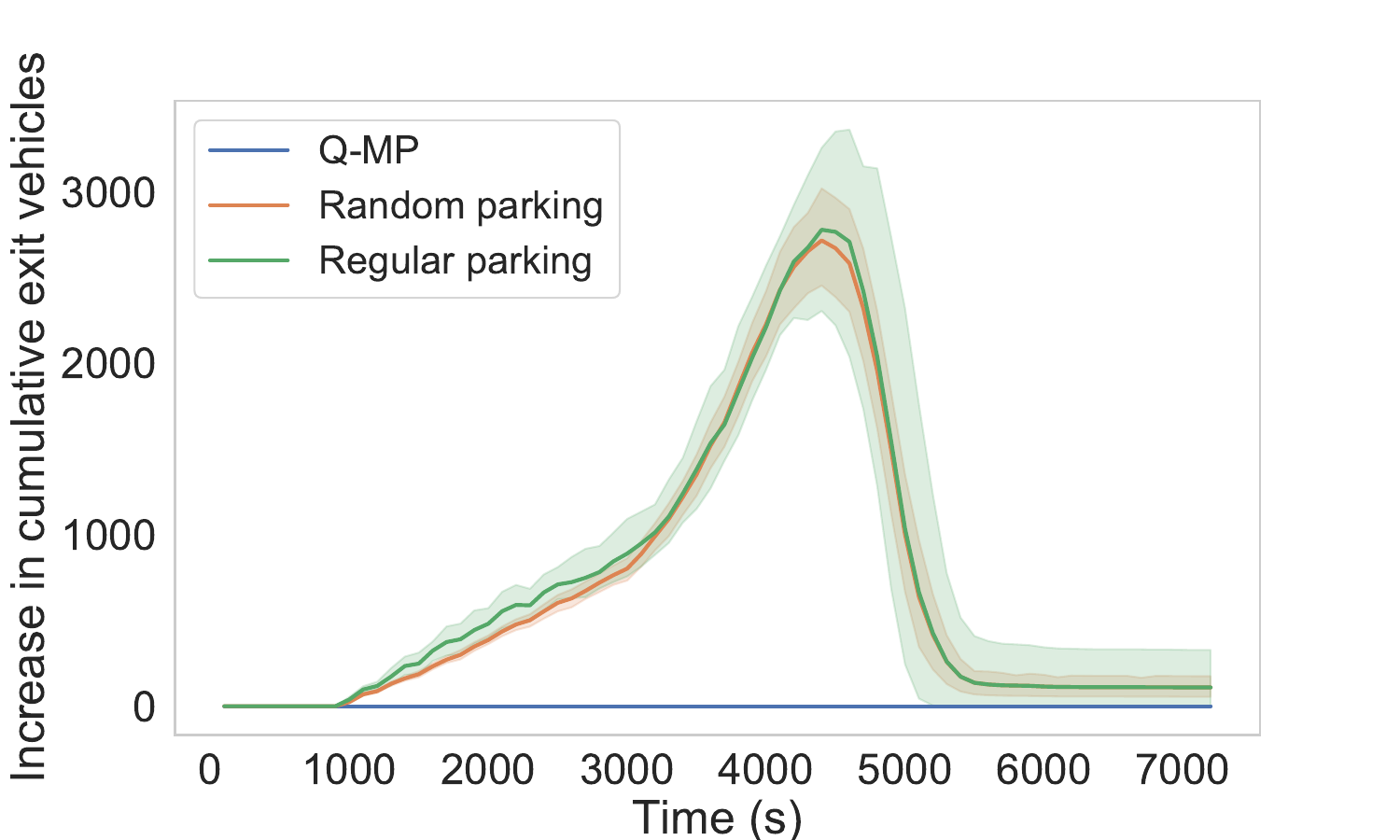}
		\caption{$i=7$}
		\label{fig:i=7}
	\end{subfigure}
	\begin{subfigure}[h]{0.49\textwidth}
		\centering
		\includegraphics[width=\textwidth]{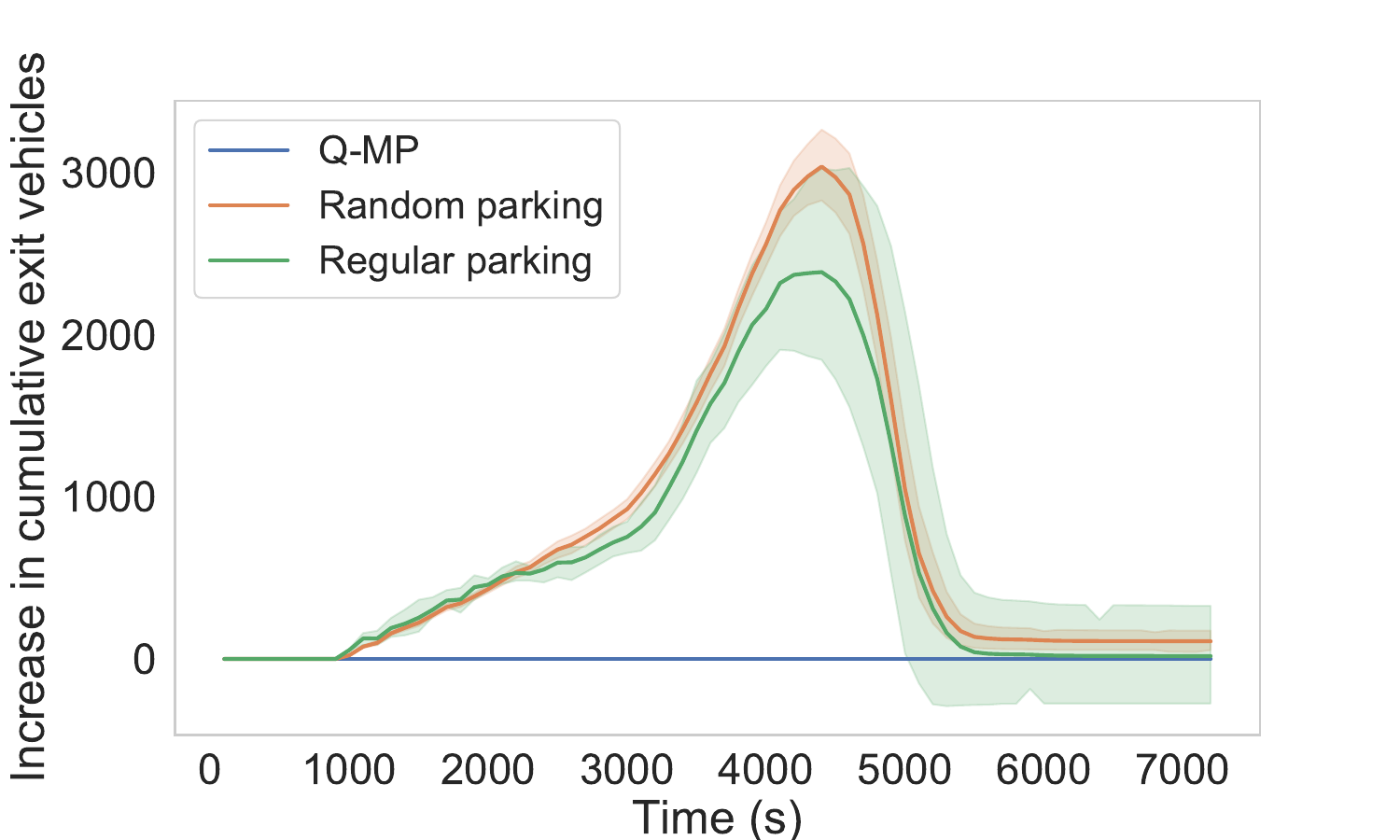}
		\caption{$i=9$}
		\label{fig:i=9}
	\end{subfigure}
	\hfill
	\caption{Influence of locations of parking locations.}
	\label{fig:impact of random parking location}
\end{figure}

Overall, Figure \ref{fig:impact of random parking location} indicates that the proposed holding strategy does not require a specific distribution of parking locations. On the contrary, the random parking locations can potentially help enhance the control performance further. 

\subsubsection{Random parking locations with parking capacity limits}\label{sec:randomparking_limit}
To examine the impact of parking capacity under random spatial distributions of parking facilities, we apply the same capacity limits considered in Section \ref{sec: results with parking limit} to the parking configuration shown in Figure \ref{fig:i=7 loc}. The corresponding results are presented in Figure \ref{fig:capacity_random}. The findings are consistent with those obtained under the perimeter-like parking distribution (see Figure \ref{fig:influence of parking limit}). Even a parking capacity of 10 spaces per facility yields substantial operational improvement compared to the baseline without the holding strategy. More importantly, a capacity limit of 20 spaces produces performance nearly identical to the scenario with no parking constraint. These results demonstrate that the proposed holding strategy does not rely on unrealistic parking availability and remains effective under more general conditions, including irregular spatial distributions and limited parking capacities.

\begin{figure}[h]
	\centering
	\includegraphics[width=0.5\textwidth]{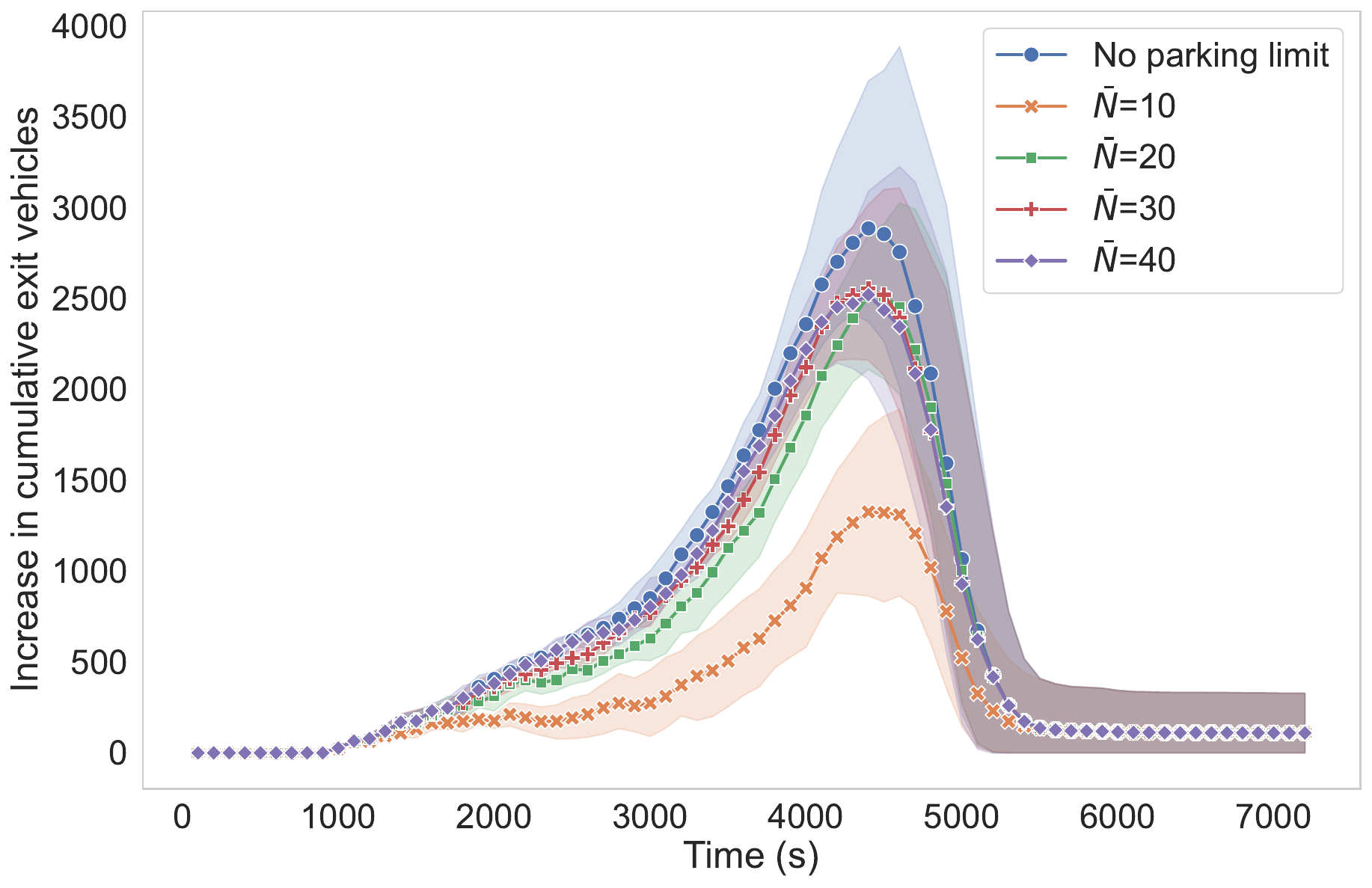}
	\caption{Impact of parking capacity on random facility distribution.}
	\label{fig:capacity_random}
\end{figure}

Moreover, to validate the design choice of holding longer trips rather than shorter ones—as theoretically motivated in Section \ref{sec:mfd holding decision}—we compare the proposed mechanism with an alternative strategy that holds vehicles whose remaining travel distance is shorter than 8 links. The comparison is conducted on the same network shown in Figure \ref{fig:i=7 loc}, under both infinite and limited parking capacity scenarios. The limited-capacity case is particularly important, as it ensures that the number of held vehicles is comparable between the two strategies. This allows us to isolate the effect of the holding rule itself, rather than differences arising from the total number of vehicles removed from the network.

\begin{figure}[h]
	\centering
	\includegraphics[width=0.5\textwidth]{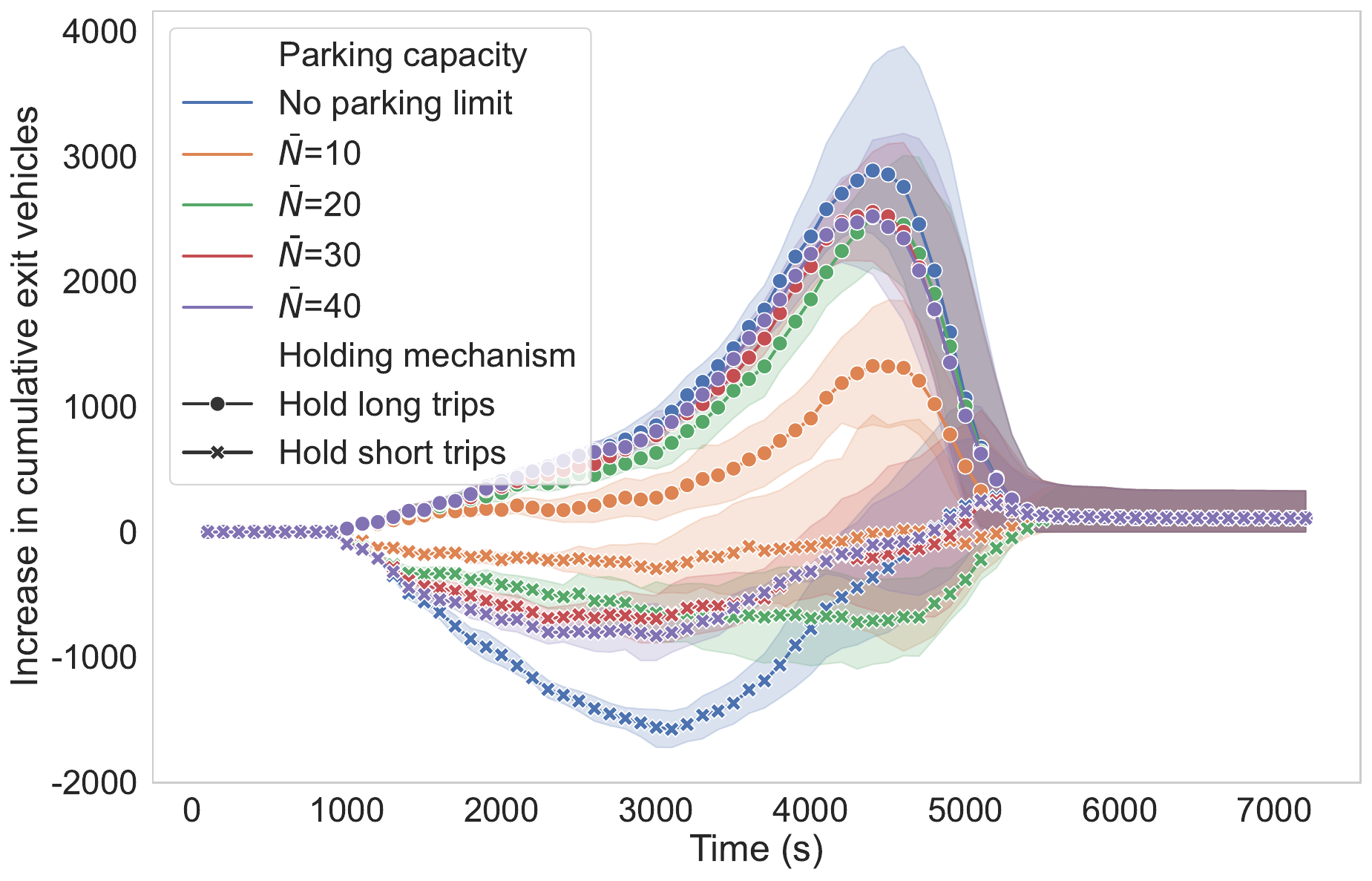}
	\caption{Comparison between holding long trips vs. holding short trips.}
	\label{fig:comparision_holding_mechanism}
\end{figure}

The results are presented in Figure \ref{fig:comparision_holding_mechanism}. Under the tested setting with $\tau=6$ mins, holding vehicles with shorter remaining distances leads to worse performance than the baseline scenario without any holding strategy. This outcome is consistent with the theoretical analysis. Holding short trips increases the average remaining travel distance of vehicles circulating in the network, thereby reducing the vehicle exit rate. Furthermore, since vehicles with longer remaining distances inherently require more time to complete their trips, re-introducing held short-trip vehicles to the network after $\tau=6$ mins may not meaningfully alleviate congestion conditions. As a result, the overall control effectiveness deteriorates.

\subsection{Results in a partial CAV environment}\label{sec: partial cav}

Previous sections assume that all vehicles in the network are CAVs so that all vehicles can be held whenever needed. Although this assumption will not be realized in the near future, the mixed environment including both CAVs and HDVs is anticipated to be the typical condition for transportation systems. Therefore, the control performance of the proposed approach is tested in such mixed environment in which only CAVs can be held. Following the previous section, 28 centroids are randomly selected to be the parking locations for CAVs. The values \{20, 40, 60, 80, 100\} are tested for $\alpha$, the penetration rate of CAVs. The results are shown in Figure \ref{fig:impact of penetration rate in a mixed traffic environment}. Overall, the control performance can be improved with the increase in the penetration rate. Both Figures \ref{fig:pr on mfd} and \ref{fig:pr on delay} show that a system with a penetration rate of $\alpha=80$ can generate similar performance with a full CAV environment. In addition, the marginal improvement when $80>\alpha>40$ is more substantial than that when $\alpha<40$. 

\begin{figure}[!ht]
	\centering
	\begin{subfigure}[h]{0.49\textwidth}
		\centering
		\includegraphics[width=\textwidth]{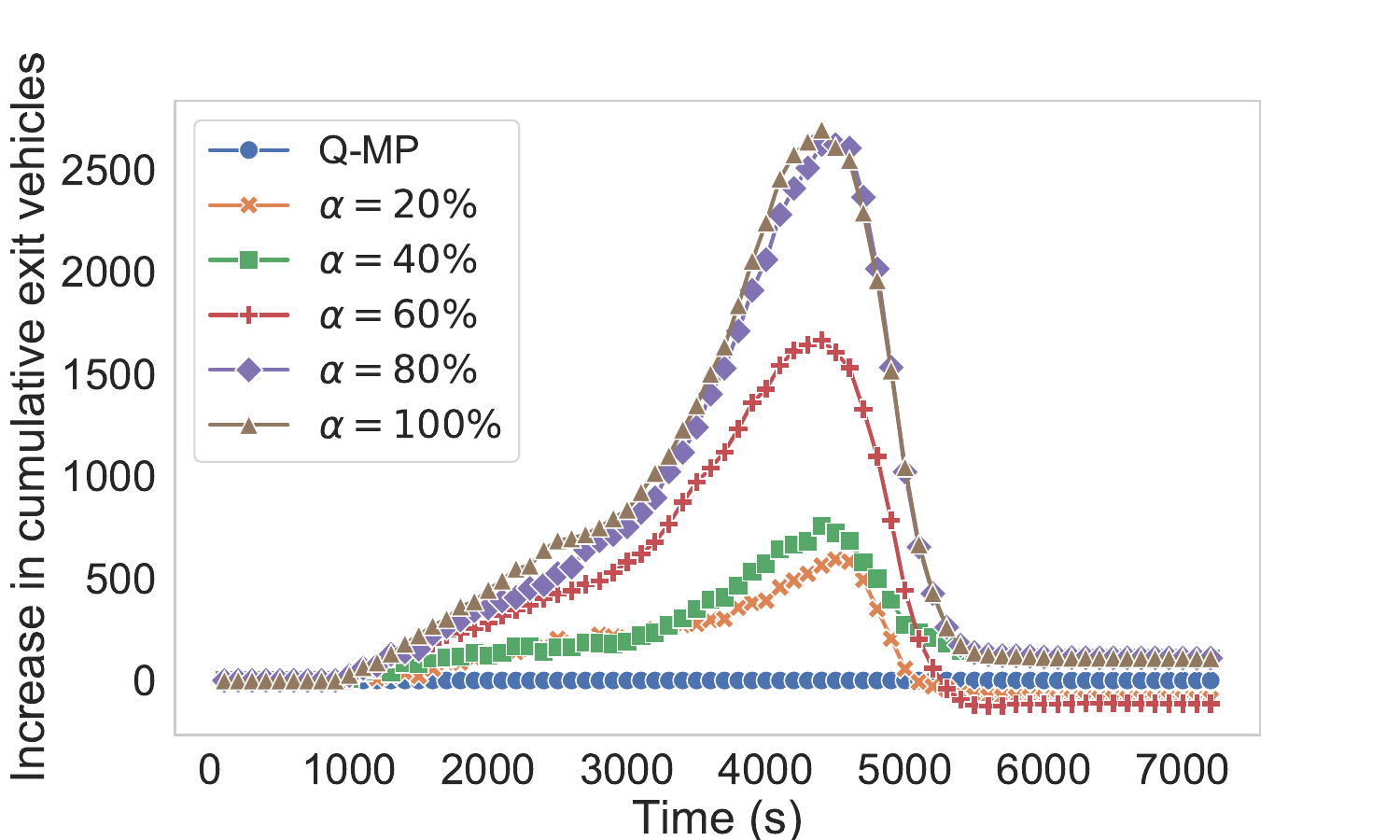}
		\caption{Comparison of exit vehicles}
		\label{fig:pr on mfd}
	\end{subfigure}
	\begin{subfigure}[h]{0.49\textwidth}
		\centering
		\includegraphics[width=\textwidth]{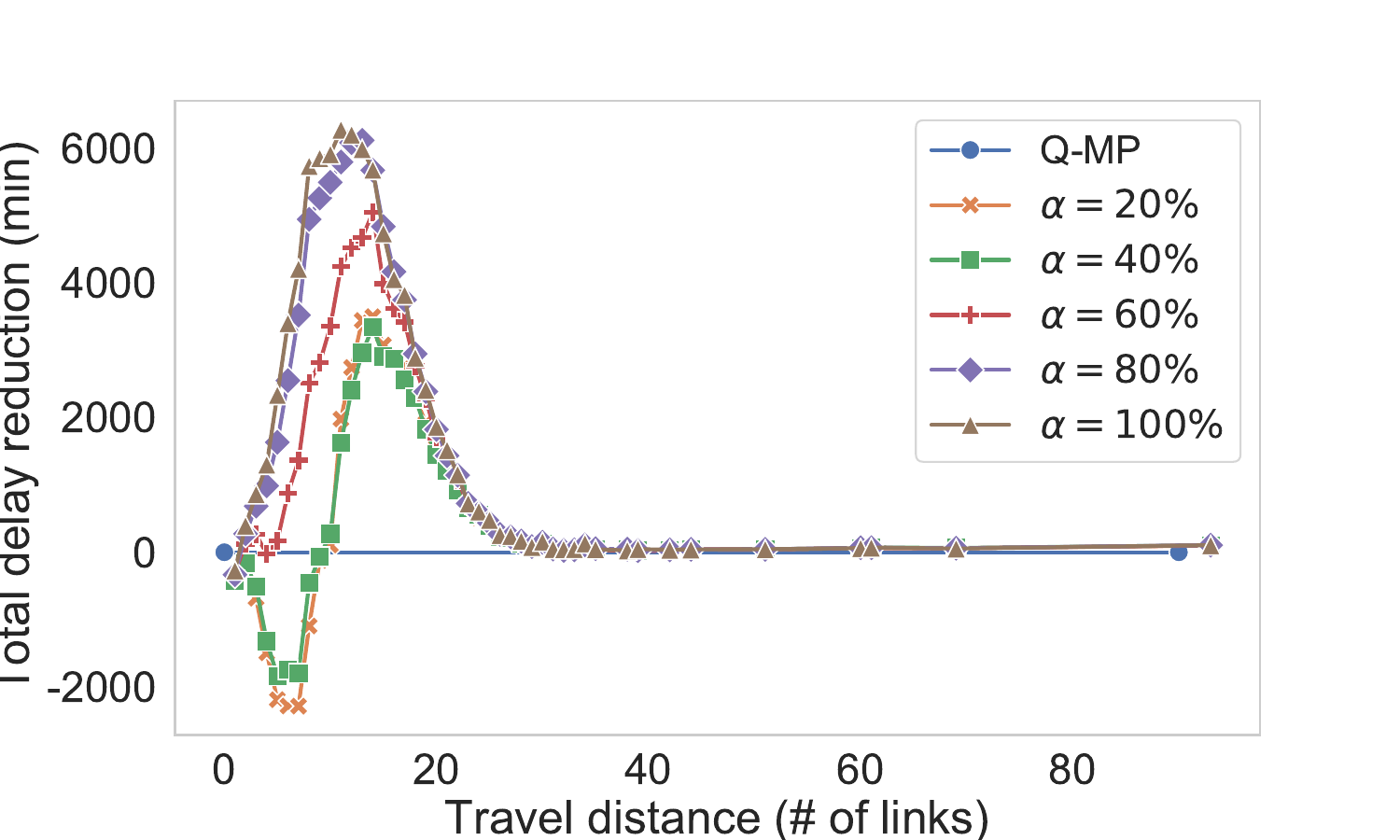}
		\caption{Vehicle delay}
		\label{fig:pr on delay}
	\end{subfigure}
	\hfill
	\caption{Impact of penetration rate in a mixed traffic environment.}
	\label{fig:impact of penetration rate in a mixed traffic environment}
\end{figure}

Note that although the HDVs are not held in the simulation, in reality, a proportion of HDVs may be willing to follow the holding suggestion sent by the control agent if the human drivers are assured that their temporary waiting can help enhance the overall control performance and even possibly reduce their own travel time. The results in this section provide insight on the required number of HDVs that follow the holding suggestion to achieve the expected control performance in that environment. 

\section{Concluding Remarks}\label{sec: conclusion}
This paper proposes a simple perimeter-free regional traffic control strategy for traffic networks with CAVs. When the network becomes congested, the proposed approach temporarily holds a subset of CAVs with long remaining travel distance at nearby parking locations to reduce the traffic burden for vehicles with short travel distance. The held vehicles will re-enter the network after a certain time period holding. By selectively holding long-distance trips, the strategy generates a self-reinforcing congestion mitigation mechanism. The temporary removal of these vehicles directly reduces network accumulation. Meanwhile, because the remaining vehicles tend to have shorter remaining distances, they can complete their trips more quickly, increasing the vehicle exit rate and further accelerating congestion dissipation. This dual effect enhances traffic system recovery beyond what would be achieved through holding arbitrary or short-trip vehicles.

Compared to the traditional perimeter control methods, which are widely used for network-level traffic control, the proposed algorithm makes use of available parking infrastructure to partially overcome the negative impact resulting from queue accumulation from the perimeter control methods and offers more robust control with respect to demand patterns. It can also be implemented without the need to define a perimeter for the region to be protected, making it more flexible. 

Micro-simulation results show that the travel time of non-held vehicles can always be reduced when the proposed strategy is enacted. However, surprising, the travel times of some held vehicles can also be reduced, leading to a significant improvement in the overall operational efficiency. The superiority of the proposed algorithm over two benchmark perimeter control methods has also been demonstrated. Moreover, the control performance has been demonstrated under various configurations of parking locations and capacities. This suggests that existing excess parking capacity - no matter how small or where located - can be leveraged to improve operational performance. The results also can be used  by policy makers to determine if to expand parking facilities. In addition, the results from a mixed environment, including both CAVs and HDVs, provide insight for policy making to encourage HDVs to follow this temporary holding strategy. 

With the benefits of the proposed method demonstrated in the findings, this paper serves as a preliminary study paving the way for future refinement of this novel traffic flow management strategy. The thresholds for both remaining travel distance and holding time are two of the essential parameters for the proposed algorithm. It is an interesting research direction to develop theoretical optimization based models for the values of both parameters.

In addition, the proposed method does not explicitly consider the interaction between reserved parking for holding and regular parking demand, which could alter traffic patterns resulting from additional cruising or unmet parking demand. Another future research direction is to jointly optimize the holding strategy and parking allocation to ensure the overall control performance, taking the influence of the holding strategy on parking demand into consideration. Moreover, in practice, incentive mechanisms such as parking fee waivers, priority re-entry, or financial compensation could improve participation. Designing dynamic incentive schemes based on trip characteristics, real-time network conditions, and expected benefits such as system travel time reduction is another important direction for future research.

\section{Acknowledgements}
This research was supported by NSF, United States Grants CMMI-1749200 and CMMI-2554235.

\printbibliography
\clearpage

\end{document}